\documentclass[12pt,reqno]{article}
\usepackage{authblk}
\usepackage[sort]{cite}
\usepackage{a4wide,url}

\usepackage[multiple]{footmisc}
\usepackage{amssymb}
\usepackage{amsmath}
\usepackage{stmaryrd}
\usepackage{amsthm}
\usepackage{xcolor} 
\usepackage{tikz}

\newcommand{\tlc}{\mathsf{Timelikecurve}}
\newcommand{\comp}{\fatsemi}

\newcommand{\Mod}{Mod}
\newcommand{\Fm}{Fm}
\newcommand{\restr}{\!\!\upharpoonright}
\newcommand{\oszt}{\slash}
\newcommand{\gyok}{\sqrt{\phantom{n}}}
\newcommand{\de}{\stackrel{\text{\tiny def}}{=}}
\newcommand{\cl}{\ensuremath{\mathsf{cl}}} 
\newcommand{\Id}{\ensuremath{\mathsf{Id}}} 
\newcommand{\defiff}{\stackrel{def}{\Longleftrightarrow}}
\newcommand{\Iff}{\; \Longleftrightarrow\;}
\newcommand{\Setclose}{\,\right\}}
\newcommand{\setmid}{\::\:}
\newcommand{\Setopen}{\left\{\,}
\newcommand{\dom}{Dom\,}
\newcommand{\ran}{Ran\,}
\newcommand{\Reals}{\mathbb{R}}
\newcommand{\Rationals}{\mathbb{Q}}
\newcommand{\ev}{\mathsf{ev}} 
\newcommand{\wl}{\mathsf{wl}} 
\newcommand{\w}{\mathsf{w}} 
\newcommand{\Ob}{\mathsf{Ob}} 
\newcommand{\Ph}{\mathsf{Ph}}
\newcommand{\Q}{\mathit{Q}}
\newcommand{\B}{\mathit{B}}
\newcommand{\W}{\mathsf{W}}
\newcommand{\Log}{\mathcal{L}}
\newcommand{\LG}{\mathcal{G}}
\newcommand{\LM}{\mathcal{M}}
\newcommand{\G}{\mathfrak{G}}
\newcommand{\M}{\mathfrak{M}}
\newcommand{\Mn}{\mathbf{M}}
\newcommand{\TeM}{\mathbf{T_eM}}
\newcommand{\gem}{\mathbf{g_e}}
\newcommand{\bvv}{\mathbf{\bar v}}
\newcommand{\bvw}{\mathbf{\bar w}}
\newcommand{\I}{\mathit{I}}
\newcommand{\g}{\mathsf{g}}

\newcommand{\vx}{\bar x}
\newcommand{\vy}{\bar y}
\newcommand{\vz}{\bar z}
\newcommand{\vv}{\bar v}
\newcommand{\vw}{\bar w}
\newcommand{\vu}{\bar u}
\newcommand{\vo}{\bar o}
\newcommand{\ve}{\bar e}
\newcommand{\vt}{\bar t}
\newcommand{\vp}{\bar p}
\newcommand{\vq}{\bar q}

\newcommand{\bv}{\ensuremath{\mathbf{v}}}

\definecolor{thmcolor}{rgb}{0,0,.4} 
\definecolor{remarkcolor}{rgb}{0,.2,0} 
\definecolor{proofcolor}{rgb}{.4,0,0} 
\definecolor{quecolor}{rgb}{.2,.2,0} 
\definecolor{axcolor}{rgb}{.23,0,.23}
\definecolor{thmbgcolor}{rgb}{0.9,0.9,1} 
\definecolor{rmbgcolor}{rgb}{0.9,1,0.9} 
\definecolor{proofbgcolor}{rgb}{1,0.9,0.9}

\newcommand{\ax}[1]{\textcolor{axcolor}{\ensuremath{\mathsf{#1}}}} 
 
\theoremstyle{definition} \newtheorem{thm}{\colorbox{thmbgcolor}{\textcolor{thmcolor}{Theorem}}}[section] 
\theoremstyle{definition} \newtheorem{cor}[thm]{\colorbox{thmbgcolor}{\textcolor{thmcolor}{Corollary}}} 
\theoremstyle{definition} \newtheorem{lemma}[thm]{\colorbox{thmbgcolor}{\textcolor{thmcolor}{Lemma}}}
\theoremstyle{definition} \newtheorem{prop}[thm]{\colorbox{thmbgcolor}{\textcolor{thmcolor}{Proposition}}}
 
\theoremstyle{remark}  
\theoremstyle{remark} 
\theoremstyle{definition}  
\theoremstyle{definition} \newtheorem{rem}[thm]{\colorbox{rmbgcolor}{\textcolor{remarkcolor}{Remark}}}
\theoremstyle{definition}

\begin{document}

\title{An Axiom System for General Relativity Complete with respect to Lorentzian Manifolds
\thanks{This research is supported by the Hungarian Scientific Research Fund
for basic research grants No.~T81188 and No.~PD84093.}}
\author[1]{H. Andr\'eka}
\author[1]{J. X. Madar\'asz} 
\author[1]{I. N\'emeti}
\author[1]{G. Sz\'ekely} 
\affil[1]{Alfr{\'ed} R{\'e}nyi Institute of Mathematics, Hungarian Academy of Sciences, Budapest, 1364 Hungary.  Emails: {\{\protect\url{andreka.hajnal}, \protect\url{madarasz.judit}, \protect\url{nemeti.istvan}, \protect\url{szekely.gergely}\}\protect\url{@renyi.mta.hu}}}

\maketitle

\begin{abstract}
We introduce several axiom systems for general relativity and show
that they are complete with respect to the standard models of general
relativity, i.e., to Lorentzian manifolds having the corresponding
smoothness properties.
\end{abstract}

\section{Introduction}
In physics, the same way as in mathematics, axioms are the basic
postulates of the theory.  However, in physics the statements are
related to the real physical world and not just to abstract
mathematical constructions. Therefore, the role of the axioms (the
role of statements that we assume without proofs) in physics is more
fundamental than in mathematics. That is why we aim to formulate
simple, logically transparent and intuitively convincing axioms. All
the surprising or unusual predictions of a physical theory should be
provable as theorems and not assumed as axioms. For example, the
prediction ``no observer can move faster than light'' is a theorem in
our approach and not an axiom, see e.g., \cite{pezsgo}, \cite[Thm
  1.]{logroad}.

In this paper, we introduce an axiom system \ax{GenRel} for general
relativity (GR) and show that it is complete with respect to the standard
models of GR, i.e., to (continuously differentiable)
Lorentzian manifolds, see Theorem~\ref{thm-c00}. This means that any
statement true in the standard models can be proved from \ax{GenRel},
see Corollary \ref{cor-c00}. Then we will generalize these results for
smooth (and $n$-times continuously differentiable) Lorentzian
manifolds, see Theorem~\ref{thm-c0} and Corollary~\ref{cor-c0}.

In GR, Einstein's field equations give the connection
between the geometry of the spacetime and the energy-matter
distribution (given by the energy-momentum tensor field). The concept
of timelike geodesic and thus all the important geometric notions of
spacetimes are definable in the models of our axioms, see
Section~\ref{sec-geod} and \cite{logroad}. 

Therefore, we can use Einstein's equations as a definition of the
energy-momentum tensor, see e.g., \cite{benda} or \cite[\S 13.1,
  p.169]{dinverno}, or we can extend the language of our geometric
theory by the concept of energy-momentum tensor and assume Einstein's
equations as axioms.  There are only methodological differences
between these two approaches. In both cases, we can assume any extra
condition about the energy-momentum tensor as a new axiom.

We follow in the footsteps of several great predecessors since logical
axiomatization of physics, especially that of relativity theory, goes
back to such leading mathematicians and philosophers as Hilbert, G{\"
  o}del, Carnap, Reichenbach, Suppes and Tarski.

Logical axiomatization of relativity theory also has an extensive
literature, see e.g., Ax \cite{ax}, Basri \cite{Basri},
Benda \cite{benda}, Goldblatt \cite{goldblatt}, Latzer \cite{Latzer},
Mundy \cite{mundy-oaomstg}, \cite{mundy-tpcomg},
Pambuccian \cite{Pambuccian}, Robb \cite{Robb}, \cite{Robb2}
Suppes \cite{suppes-sopitposat}, Schutz \cite{schutz},
\cite{schutz-aasfmst}, \cite{Schu},  Szab{\'o} \cite{Szabo}.

Our goals go beyond the earlier approaches in several aspects. For
example, we not searching for a single monolithic axiom system, but we
are building a whole flexible hierarchy of axiom systems. We also make
extra effort to get a deep understanding of the connections between
the elements of this hierarchy, see e.g., \cite{logroad}, and the
relations between axiom systems formulated using different basic
concepts, see e.g., \cite{AN-CompTh}, \cite{Mphd}.

Another novelty in our approach is that we concentrate on the
transition from special relativity (SR) to GR, we try to keep this
transition logically transparent and illuminating even for the
non-specialists.  Starting from our streamlined axioms system
\ax{SpecRel} of SR, we can ``derive'' the axioms of \ax{GenRel} in two
natural steps, see \cite{logroad}. The axioms of \ax{GenRel} are
basically the localized versions of the axioms (and some theorems) of
\ax{SpecRel}.

The success story of using axiomatic method and foundational thinking
in the foundations of mathematics also enforces our firm belief that
it worth to apply them in the foundations of spacetime theories, see
also Harvey Friedman~\cite{FriFOM1}, \cite{FriFOM2}.

For good reasons, foundations of mathematics was carried through
strictly within first-order logic (FOL).  For the same reasons,
foundations of spacetime theories are best developed within FOL. For
example, in any foundational work it is essential to avoid tacit
assumptions, and one acknowledged feature of using FOL is that it
helps to eliminate tacit assumptions.  There are several further
reasons why we work within FOL, see \cite[\S Why FOL?]{pezsgo},
\cite[\S 11]{Szphd}.

\section{Axioms for General Relativity}

First, we introduce the basic concepts of our FOL axiom system
\ax{GenRel} for GR. We are going to consider two sorts of objects
mathematical and physical.  Mathematical objects will be called
\textbf{quantities}, they will represent physical quantities, such as
speeds or coordinates. We include addition, multiplication and
ordering as basic concepts on quantities. Physical objects will be
called \textbf{bodies}.  We will associate a body ``sitting'' at the
origin to every coordinate system. We will call these bodies
\textbf{observers}.  Light signals (\textbf{photons}) will be another
special type of bodies our axioms will speak about.  Coordinate
systems will be represented by one relation $\W$ that we will call
\textbf{worldview relation}; $\W(m,b,x_1,\ldots,x_d)$ means
intuitively that ``observer $m$ coordinatizes body $b$ by coordinates
$x_1,\ldots,x_d$ (in his coordinate system).'' Here, $d$ is a fixed
natural number determining the dimension of the coordinate
systems.\footnote{The fact that all coordinate systems are represented
  by one relation implies that they all have the same dimensions. See
  \cite{mythes} for a similar axiomatic approach in which the
  dimension of coordinate systems is observer dependent.}

The above means that we will use the following formal FOL
  language for axiomatizing GR:
\begin{equation*}
\{\, \B, \Q, +,\cdot, \le,\Ph,\Ob, \W\,\},
\end{equation*}
where $\Q$ is a sort for quantities; $\B$ is a sort for bodies;
  $+,\cdot$ are binary operations of sort $\Q$ and $\le$ is a binary
  relation of sort $\Q$. $\Ob$ and $\Ph$ are unary relations of sort
  $\B$ for observers and photons; finally, $\W$ is a $2+d$-place
  relation connecting $\B$ and $\Q$ (the first two arguments are of
  sort $\B$ and the rest are of sort $Q$).  More about the intuition
  and the why behind our choosing of this language can be found, e.g.,
  in \cite[\S 2]{logroad}.

Now we are ready to list the axioms of \ax{GenRel}. The first
  axiom provides some useful and widely used properties of real
  numbers for the quantities.  
\begin{description}
\item[\underline{\ax{AxEField}}] The structure $\langle
  \Q,+,\cdot,\le\rangle$ of quantities is a Euclidean field, i.e.,\\
$\bullet$ $\langle\Q,+,\cdot\rangle$ is a field in the sense of abstract
algebra;\footnote{The field-axioms (see e.g., \cite[pp.40--41]{CK},
\cite[p.38]{Hodges}) say that $+$, $\cdot$ are associative and
commutative, they have neutral elements $0$, $1$ and inverses $-$,
$\oszt$ respectively, with the exception that $0$ does not have an
inverse with respect to $\cdot\,$, as well as $\cdot$ is additive with
respect to $+$.}\\
$\bullet$ the relation $\le$ is a linear ordering on $\Q$ such that  
\begin{itemize}
\item[i)] $x \le y\rightarrow x + z \le y + z$ and 
\item[ii)] $0 \le x \land 0 \le y\rightarrow 0 \le xy$
holds; and
\end{itemize}
$\bullet$ nonnegative elements have square roots: $0\le x \rightarrow \exists y\enskip x=y^2$.
\end{description}
 We will use $0$, $1$, $-$, $\oszt$, $\gyok$ as derived (i.e.,
 defined) operation symbols.  

\ax{AxEField} is sufficient in SR for proving the main
  predictions; however, in GR we will have to use more properties of
  real numbers, see axiom schema \ax{CONT} on p.\,\pageref{p-cont}.

The next two axioms speak about the so-called worldviews of
observers. The \textbf{worldline} of body $b$ according to observer
$m$ is defined as the collection of those coordinate points where $m$
coordinatizes $b$, i.e.,
\begin{equation*}
\wl_m(b)\de\{\, \vx: \W(m,b,\vx)\,\},
\end{equation*}
where $\vx$ abbreviates $n$-tuple $\langle x_1,\ldots,x_n\rangle$.

In SR, the worldlines of photons are straight lines, while in GR these
worldlines are more general curves. The notion of velocity for these
curves is the velocity of their straight line approximations.  Our
central axiom for GR will state that the velocity of a photon is 1
according to an observer when meeting it.  To introduce this axiom, we
need some definitions and notations.

In our formulas, we will use the usual logical connectives $\lnot$
(\textit{not}), $\land$ (\textit{and}), $\lor$ (\textit{or}),
$\rightarrow$ (\textit{implies}), $\leftrightarrow$
(\textit{if-and-only-if}) and FOL quantifiers $\exists$
(\textit{exists}) and $\forall$ (\textit{for all}).

In order to define velocity for curved worldlines, let us introduce a
concept of approximation. Let $f,g:\Q^m\to \Q^n$, $m,n\ge 1$ be
partial\footnote{Partial means that $f$ and $g$ are not necessarily
  everywhere defined on $\Q^m$.}  maps and $\vx\in\Q^m$. We say that
\textbf{$f$ approximates $g$ at $\vx$}, in symbols $f\sim_{\vx} g$, if
\begin{multline*}
\forall \varepsilon >0\enskip\exists \delta>0\enskip\forall \vy\enskip
\big(|\vy-\vx|\le\delta\\\rightarrow \vy\in\dom f \cap \dom g\land
|f(\vy)-g(\vy)|\le\varepsilon\cdot |\vy-\vx|\big),
\end{multline*}
where $\dom f$ is the domain of function $f$ (see
p.\,\pageref{p-domain}) and the {\bf Euclidean length} $|\vz|$ of
$\vz\in\Q^k$ is defined as $\sqrt{z_1^2+\ldots+z^2_k}$.

\begin{rem}\label{rem-conv}
By its definition, $f\sim_{\vx} g$ implies that $\vx$ has an open
neighborhood where both $f$ and $g$ are defined; and that
$f(\vx)=g(\vx)$.  Approximation at a given point is an equivalence
relation on functions; and if two affine maps (i.e., maps that are
composition of translations and linear maps) approximate each other,
then they are equal. These facts can easily be proved from
$\ax{AxEField}$.
\end{rem}

When $f$ is a unary  function, i.e., when $m=1$ above, the
 notion of derivative%
\footnote{The derivative of $f$ is usually defined as the limit
  $\lim_{h\rightarrow 0}\frac{f(x+h)-f(x)}{h}$, this is equivalent to
  our definition. Its intuitive meaning is how fast and in which
  direction the function increases at $x$.} can be defined by the
above concept of approximation as:
\begin{equation*}
f'(x)=\vy \defiff f\sim_{x} \{ \langle x+t, f(x)+t\cdot\vy\rangle : t\in\Q\}.
\end{equation*}
By this definition, the derivative of $f$ at $x$ is an $n$-dimensional
vector, we call it the \textbf{derivative vector} of $f$ at
$x$.\label{DerivativeVector}

It will be convenient to use the notions of \textbf{space component}
and \textbf{time  component} of $\vx\in\Q^d$, respectively:
\begin{equation*}
\vx_s\de \langle x_1,x_2,\ldots, x_{d-1}\rangle \quad \text{ and
}\quad {x_t}\de x_d.
\end{equation*}

Assume that the worldline $\wl_m(b)$ of body $b$ is a function of time
and $\dom \wl_m(b)$ is open (i.e., $\forall \vx,\vy\in\wl_m(b)\enskip
[x_t=y_t\to \vx_s=\vy_s]$ and $\forall \vx\in \wl_m(b)\;\exists
\delta>0\; \forall t\;[|t-x_t|<\delta\to \exists \vy\in \wl_m(b)\;
  y_t=t]$).\footnote{ To abbreviate formulas, we use bounded
  quantifiers in the following way: $\exists x\; [\varphi(x)\land
    \psi]$ and $\forall x\; [\varphi(x)\rightarrow \psi]$ are
  abbreviated to $\exists x\in\varphi\enskip \psi$ and $\forall
  x\in\varphi\enskip \psi$, respectively. So $\forall\vx\vy\;
  [\W(m,b,\vx)\land\W(m,b,\vy) \rightarrow \psi]$ is abbreviated to
  $\forall \vx,\vy\in\wl_m(b) \enskip \psi$.}\footnote{Both
  $\vx\in\wl_m(b)$ and $b\in\ev_m(\vx)$ below represent the same
  atomic formula of our FOL language, namely: $\W(m,b,\vx)$.}  Then
the velocity of body $b$ according to observer $m$ at $\vx\in\wl_m(b)$
is defined as the time-derivative of the worldline of $b$ at $x_t$:
\begin{equation*}
\bv_m(b,\vx)\de \wl_m(b)'(x_t).
\end{equation*}

 We defined velocity $\bv_m(b,\vx)$ only if $\vx\in\wl_m(b)$ and
 $\wl_m(b)$ is a function of time defined at an open interval
 containing $x_t$. Let us denote these assumptions by
 $\vx\in\dom\bv_m(b)$. Now we are ready for formulating the central
 axiom of \ax{GenRel}:

\begin{description}
\item[\underline{\ax{AxPh^-}}] The speed of a photon an observer
  ``meets'' is 1 when they meet, and it is possible to send out a
  photon in each direction where the observer stands:
\begin{multline*}
  \forall m\in\Ob\enskip \forall p\in\Ph\enskip \forall\vx\; \big[\W(m,m, \vx)\land\W(m,p,\vx) \rightarrow
\\ \vx\in\dom\bv_m(p)\land |\bv_m(p,\vx)|=1\big]\text{, and}
\end{multline*}
\begin{multline*}
\forall m\in\Ob\enskip \forall
\vx\bv\;\big(\W(m,m,\vx)\land |\bv|=1 \rightarrow
\hfill\\\exists p\in\Ph\; \big[\W(m,p,\vx) \land
  \vx\in\dom\bv_m(p)\land \bv_m(p,\vx)=\bv\big]\big).
\end{multline*}
\end{description}

The next axiom talks about the worldlines of observers. Let $\vo$
denote the origin of $\Q^{d-1}$, i.e., $\vo\de \langle
0,\ldots,0\rangle$.

\begin{description}
\item[\underline{\ax{AxSelf^-}}] In his own worldview, the worldline
  of any observer is an interval of the time-axis containing all the
  coordinate points of the time-axis where the observer coordinatizes
  something:
\begin{multline*}
\forall m\in\Ob\enskip\forall \vx\in\wl_m(m)\enskip \vx_s=\vo, \text{
  and }\\  \forall
m\in\Ob \enskip\forall \vx,\vy\in\wl_m(m)\enskip \forall t\;\big[
  x_t<t<y_t\rightarrow \W(m,m,\vo,t)\big], \text{ and }\\ \forall
m\in\Ob\enskip \forall t \;\big[\exists b \; \W(m,b,\vo,t) \rightarrow
  \W(m,m,\vo,t)\big].
\end{multline*}
\end{description}

By the \textbf{event} occurring for observer $m$ at coordinate point
$\vx$, we mean the set of bodies $m$ coordinatizes at $\vx$:
\begin{equation*}
\ev_m(\vx)\de\{\, b : \W(m,b,\vx)\,\}.
\end{equation*}

\begin{description}
\item[\underline{\ax{AxEv^-}}] Observers see all the events in which they participate:
\begin{equation*}
\forall mk\vx\big(\Ob(k)\land \W(m,k,\vx)\rightarrow\exists \vy\enskip
\ev_m(\vx)=\ev_k(\vy)\big).\footnote{ $\ev_m(\vx)=\ev_k(\vy)$ is an
  abbreviation of formula $\forall b \;[ \W(m,b,\vx)\leftrightarrow
  \W(k,b,\vy)]$.}
\end{equation*}
\end{description}

It is  convenient to introduce the \textbf{worldview transformation}
between observers $m$ and $k$ as the binary relation connecting
those coordinate points in which $m$ and $k$ see the same nonempty
events:
\begin{equation*}
\w_{mk}(\vx,\vy)\defiff \ev_m(\bar x)=\ev_k(\vy)\neq\emptyset.
\end{equation*}
We regularize worldview transformations by the following axiom.
\begin{description}\label{axdiff}
\item[\underline{\ax{AxCDiff}}] The worldview transformations between
  observers are functions having linear approximations $A_{\vx}$ at each
  coordinate point $\vx$ of their domain and this linear approximation
  $A_{\vx}$ depends continuously on point $\vx$ (i.e., they are
  continuously differentiable maps):
  \begin{multline*}\forall m,k\in\Ob\;\big[\w_{mk} \text{ is a function\footnotemark\, } \land \\ \forall \vx\in\dom\w_{mk}\enskip\exists
\text{ affine map\footnotemark }\enskip A_{\vx} \enskip \w_{mk}\sim_{\vx} A_{\vx}\big]\text{, and }
  \end{multline*}
\begin{equation*}
\forall m,k\in\Ob\; \forall \varepsilon >0\enskip \exists \delta>0 \enskip \forall
\vy\vz\in\dom\w_{mk}\;\big( |\vy-\vz|<\delta \to \left|
A_{\vy}-A_{\vz}\right|<\varepsilon\big).\footnotemark
\end{equation*}
\end{description}
\begin{rem}\label{rem-ctc}
The physical meaning of that the worldview transformations are
functions is that no observer coordinatizes an event twice, i.e.,
$\forall m\in\Ob\; \forall \vx\vy \;[\ev_m(\vx)=\ev_m(\vy)\rightarrow
\vx=\vy$]. 
\end{rem}
\begin{rem}
Let us note that, by the definition of $\sim_{\vx}$, \ax{AxCDiff}
implies that the domain $\dom \w_{mk}$ of worldview transformation
$\w_{mk}$ is an open set.
 Therefore, \ax{AxSelf^-} and \ax{AxCDiff} imply that the worldline
of observer $m$, according to him, is an open interval of the
time-axis since it is the intersection of $\dom \w_{mm}$ and the
time-axis.
\end{rem}
\addtocounter{footnote}{-2} \footnotetext{That $\w_{mk}$ is a function
  can be formalized as follows: $\forall \vx\vy\vz \,[
    \w_{mk}(\vx,\vy)\land \w_{mk}(\vx,\vz) \rightarrow \vy=\vz]$.}
\addtocounter{footnote}{1}\footnotetext{The quantifier ``$\exists
  \text{ affine map } A$'' looks like a second-order logic one, but
  truly it is a FOL quantifier because every affine map from $\Q^d$ to
  $\Q^d$ can be represented by a $d\times d$ matrix and a vector of
  $\Q^d$, i.e., $d^2+d$ elements of $\Q$.}  \addtocounter{footnote}{1}
\footnotetext{By Remark~\ref{rem-conv}, affine map $A_{\vx}$
  approximating $\w_{mk}$ at $\vx$ is unique. This fact justifies that
  we can use it as a defined concept in this formula. Also the
  Euclidean distance of $A_{\vy}$ and $A_{\vz}$ is meaningful since
  $A_{\vy},A_{\vz}\in\Q^{d^2+d}$.}

Our next axiom states that the derivative of worldview transformations
are continuous also in the sense that the difference how they distort
the Minkowski metric is small for observers in close enough events. To
formulate this axiom, we have to recall some definitions.  The
\textbf{Minkowski metric} $\mu$ is defined as:
\begin{equation*}
\mu(\vv,\vw)\de v_t\cdot w_t-v_2\cdot w_2-\ldots-v_{d-1}\cdot
w_{d-1}
\end{equation*}
for all $\vv,\vw\in\Q^d$.  The \textbf{derivative}\label{pdiff} (or
linear approximation) of map $f$ at $\vx\in\Q^n$, denoted by
$[d_{\vx}f]$, is defined as follows:
\begin{equation*}
[d_{\vx}f](\vy)=A(\vy+\vx)-A(\vx) \defiff f\sim_{\vx} A\mbox{ and }
A\mbox{ is affine}.
\end{equation*}
In the case of unary functions, the connection between this notion of
derivative and derivative vector introduced at
p.\,\pageref{DerivativeVector} is the following: $f'(x)=[d_xf](1)$ and
$[d_xf](t)=t\cdot f'(x)$ for all $t\in \Q$.

\begin{description}
\item[\underline{$\ax{AxC^0\g_m}$}] The difference how the linear
  approximations of worldview transformations distort the Minkowski
  metric is small for observers in close enough events:
\begin{multline*}
\forall m\in\Ob \enskip \forall \vx \in \dom \w_{mm}\enskip \forall
\varepsilon>0 \enskip \exists\delta >0\enskip\\ \forall \vy\enskip
\forall k,h\in\Ob \enskip \Big( |\vx-\vy|<\delta \land
\W(m,k,\vx)\land \W(m,h,\vy) \to \forall \vv\vw
\\ \left|\mu\big([d_{\vx}\w_{mk}](\vv),[d_{\vx}\w_{mk}](\vw)\big)
-\mu\big([d_{\vy}\w_{mh}](\vv),[d_{\vy}\w_{mh}](\vw)\big) \right|<\varepsilon\Big).
\end{multline*}
\end{description}

The behavior of observer $k$'s
clock as seen by observer $m$ is defined as follows:
\begin{equation*}
\cl_{mk} \de \{ \langle x_t,y_t\rangle : \w_{mk}(\vx,\vy)\text{ and }
\vx\in\wl_m(k)\}.
\end{equation*}
If $\cl_{mk}$ is a function, then it is differentiable (by our
previous axioms) and $\cl_{mk}(t)$ is the time $k$'s clock shows
``when'' $m$'s clock shows $t$. Thus, e.g., $\cl_{mk}'(t)=2$ means
that at $t$ (according to $m$'s clock) $k$'s clock runs twice as fast
as $m$'s.
\begin{description}
\item[\underline{$\ax{AxSymt^-}$}] Meeting observers see each other's
  clocks slow down with the same rate:
\begin{multline*}
\forall mk\in\Ob \;\big( \cl_{mk} \text{ is a function }
\land\\ \forall \vx\vy \;\big[ m,k\in\ev_m(\vx)=\ev_k(\vy)\rightarrow
\cl_{mk}'(x_t)=\cl_{km}'(y_t)\big]\big).
\end{multline*}
\end{description}

So far, we have not assumed the existence of any observer.  By the
next axiom, we assume the existence of some slowly moving observers in
every (nonempty) event. For a more delicate assumption ensuring the
existence of an observer on every definable timelike curve segment,
see axiom schema \ax{COMPR} on p.\,\pageref{compr}.

\begin{description}
\item[\underline{\ax{AxThExp^-_{00}}}] There is an observer in every
  nonempty event:
\begin{equation*} 
\exists h\; \Ob(h)\land \forall m\in\Ob\; \exists b \big[\W(m,b,\vx)\to \exists k\;
  \Ob(k)\land\W(m,k,\vx)\big].
\end{equation*}
\end{description}

If the number line has some definable gaps, some key predictions of
relativity, such as the twin paradox my not hold, see \cite{twp},
\cite[Thms.~7.1.1 and 7.1.3]{Szphd}. Our next assumption is an axiom
excluding these gaps.

Let $\Log$ be a many sorted language containing sort $\Q$ and binary
relation $\le$ on $\Q$.
\begin{description}
\label{p-cont}
\item[\underline{\ax{CONT_\Log}}] Every subset of $\Q$ which is $\Log$-definable,
  bounded and nonempty has a supremum (i.e., least upper bound) with respect to $\le$.
\end{description}
See p.\,\pageref{p-detailedcont} for a detailed introduction of
\ax{CONT_\Log}.  Let $\LG$ be the language of \ax{GenRel}, i.e., $\LG\de
\{\B,\Q,+,\cdot,\le,\Ph,\Ob, \W\}$.

Let us now introduce an axiom systems for GR as the collection of the axioms above: 
\begin{multline*}
\ax{GenRel}\de
\ax{AxEField}+\ax{AxPh^-}+\ax{AxSelf^-}+\ax{AxEv^-}+\ax{AxCDiff}\\+\ax{AxC^0\g_m}+\ax{AxSymt^-}+\ax{AxThExp^-_{00}}+\ax{CONT_\LG}
\end{multline*}

\section{An axiomatic Theory of Lorentzian Manifolds}

Here, we introduce a FOL axiom system \ax{LorMan} of Lorentzian
manifolds, see \cite[\S 2]{wald}, \cite[\S 2.2]{GlobLorGeom} for some
non FOL definition of Lorentzian manifolds. The language of
$d$-dimensional Lorentzian manifolds is the following:
\begin{equation*}
\{\,\I, \Q, +,\cdot, \le, \psi,\g\,\}
\end{equation*}
where $\I$ (indexes) and $\Q$ (quantities) are two sorts, $+$ and
$\cdot$ are two-place function symbols of sort $\Q$, $\le$ is a
two-place relation symbol of sort $\Q$, $\psi$ (transition relation)
is a $2d+2$-place relation symbol the first two arguments of which are
of sort $\I$ and the rest are of sort $\Q$, and $\g$ (metric relation)
is a $3d+2$-place relation symbol the first argument of which is of
sort $\I$ and the rest are of sort $\Q$.


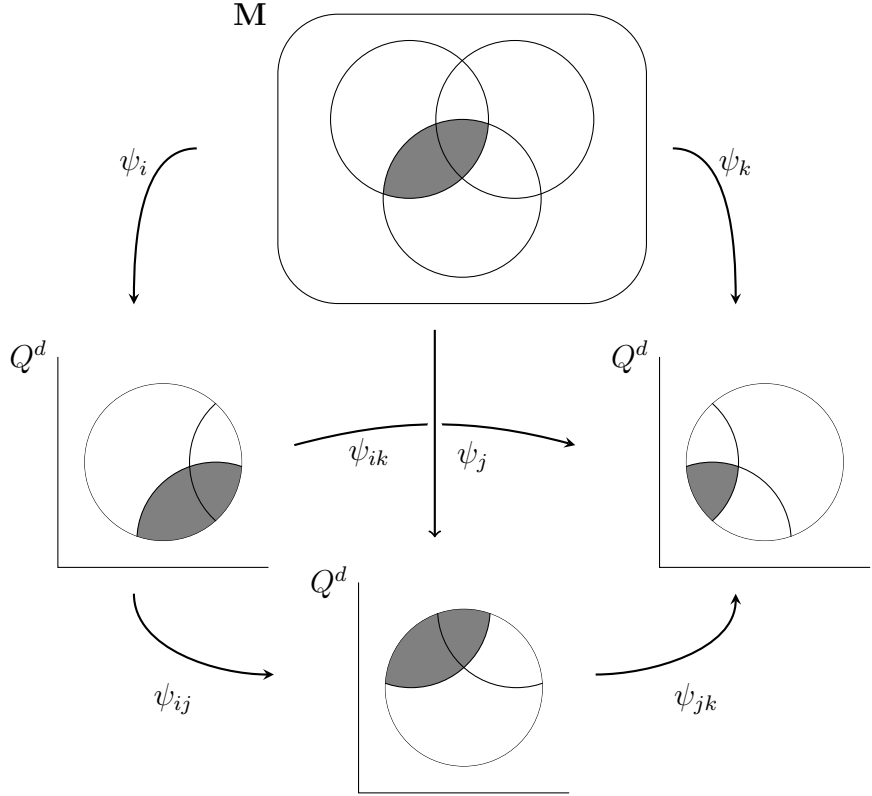
\begin{figure}
\begin{center}
\usetikzlibrary{positioning}
\begin{tikzpicture}[scale=1]

\tikzset{marrow/.style={thick, ->,shorten >=0.2cm, shorten <=0.2cm, >=stealth}}

\node (M) at (0,4){\begin{tikzpicture}[scale=0.7]
\node[left] at (-3.5,2) {$\mathbf{M}$};
\draw[rounded corners=8mm] (-3.5,-3.5) rectangle (3.5,2);
\begin{scope}
\clip (0,-1.5) circle (1.5);
\fill[gray]  (-1,0) circle (1.5);
\end{scope}
\draw (-1,0) circle (1.5);
\draw (1,0) circle (1.5);
\draw (0,-1.5) circle (1.5);
\end{tikzpicture}
};

\node (i) at (-4,0){\begin{tikzpicture}[scale=0.7]
\draw (-3,2) node[left] {$Q^d$} -- (-3,-2) -- (1,-2);
\clip (-1,0) circle (1.5);
\begin{scope}
\clip (0,-1.5) circle (1.5);
\fill[gray]  (-1,0) circle (1.5);
\end{scope}
\draw (-1,0) circle (1.5);
\draw (1,0) circle (1.5);
\draw (0,-1.5) circle (1.5);
\end{tikzpicture}
};

\node (j) at (0,-3) {\begin{tikzpicture}[scale=0.7]
\draw (-2,0.5) node[left] {$Q^d$} -- (-2,-3.5) -- (2,-3.5);
\clip(0,-1.5) circle (1.5);
\begin{scope}
\clip (0,-1.5) circle (1.5);
\fill[gray]  (-1,0) circle (1.5);
\end{scope}
\draw (-1,0) circle (1.5);
\draw (1,0) circle (1.5);
\draw (0,-1.5) circle (1.5);
\end{tikzpicture}
};

\node  (k) at (4,0) {\begin{tikzpicture}[scale=0.7]
\draw (-1,2) node[left] {$Q^d$} -- (-1,-2) -- (3,-2);
\clip (1,0) circle (1.5);
\begin{scope}
\clip (0,-1.5) circle (1.5);
\fill[gray]  (-1,0) circle (1.5);
\end{scope}
\draw (-1,0) circle (1.5);
\draw (1,0) circle (1.5);
\draw (0,-1.5) circle (1.5);
\end{tikzpicture}
};

\draw[marrow] (M.west) to [out=180,in=90] node (pi) {} (i.north);
\node[above=0.1cm of pi.north west] {$\psi_i$};
\draw[marrow] (M.east) to [out=0,in=90]  node (pk) {}(k.north);
\node[above=0.1cm of pk.north east]  {$\psi_k$};
\draw[marrow] (i.east) to [out=15,in=165]  node (ik){}(k.west); 
\fill[white] (0,0.3) circle (0.1);
\draw[thick,shorten >=0.1cm,shorten <=0.2cm, ->] (M.south) to node (pj) {}(j.north);
\node[right=0cm of pj.south east,yshift=-2mm]  {$\psi_j$};
\node[left=0.3cm of ik.south west, yshift=-2mm]  {$\psi_{ik}$};
\draw[marrow] (i.south) to [out=270,in=180]  node (ij) {} (j.west);
\node[below=0.1cm of ij.south west]  {$\psi_{ij}$};
\draw[marrow] (j.east) to [out=0,in=270] node (jk){} (k.south);
\node[below=0.1cm of jk.south east]  {$\psi_{jk}$};
\end{tikzpicture}
\caption{\label{fig-lorman} Illustration for manifold $\Mn$ and
  transition maps}
\end{center}
\end{figure}


Now we are ready to formulate the axioms of \ax{LorMan}.

\begin{description}
\item[\ax{AxFn}] The transition and the metric relations are
  functions in their last variables:
\begin{equation*}
\forall ij \vx\vy\vy' \;
\psi(i,j,\vx,\vy)\land \psi(i,j,\vx,\vy')\rightarrow \vy=\vy', \text{ and}
\end{equation*}
\begin{equation*}
\forall i\vx\vv\vw aa'\;
\g(i,\vx,\vv,\vw,a)\land \g(i,\vx,\vv,\vw,a')\rightarrow a=a'.
\end{equation*}
\end{description}

By axiom \ax{AxFn}, we can speak about the \textbf{transition map}
$\psi_{ij}$ and \textbf{metric} $\g_i$ in the following sense:
\begin{eqnarray}
\psi_{ij}(\vx)=\vy &\defiff& \psi (i,j,\vx,\vy)\enskip \text{ and
}\\ \g_i(\vx)(\vv,\vw)=a&\defiff& \g(i,\vx,\vv,\vw,a).
\end{eqnarray}
We will refer to the first and the second parts of \ax{AxFn} as
\ax{AxFn\psi} and \ax{AxFn\g}, respectively.

We think of functions as special binary relations. Hence we compose
them as relations.  The \textbf{composition} $R \comp S$ of binary
relations $R$ and $S$ is defined as:
\begin{equation*}\label{rcomp}
{R \comp S}\de \Setopen\langle a,c\rangle: \exists b\enskip R(a,b)\land  S(b,c) \Setclose.
\end{equation*}
So $(g\comp f)(x)=f\big(g(x)\big)$ if $f$ and $g$ are functions. We
will also use the notation $x\comp g\comp f$ for $(g\comp f)(x)$
because it is easier to grasp. In the same spirit, we will sometimes
use the notation $x\comp f$ for $f(x)$.  

The \textbf{domain} $\dom R$ and the \textbf{range} $\ran R$ of a
binary relation $R$ are defined respectively as:
\begin{equation*}\label{p-domain}
 \dom R \de \{\, x : \exists y\enskip R(x,y) \,\}\enskip \text{ and }\enskip \ran R \de \{\, x :
    \exists x\enskip R(x,y) \,\}.
\end{equation*}
The \textbf{inverse} of $R$ is defined as:
\begin{equation*}\label{rinv}
{R^{-1}}\de \Setopen\langle a,b\rangle:  R(b,a) \Setclose.
\end{equation*} 
Let $\Id_H$ denote the \textbf{identity map} from $H\subseteq \Q^d$ to
$H$, i.e., 
\begin{equation*}
\Id_H(\vx)=\vx \enskip \text{ for all }\enskip\vx\in H.
\end{equation*}

\begin{description}
\item[\underline{\ax{AxCom\psi}}] The transition maps satisfy the
  following basic compatibility relations:
\begin{eqnarray}
\forall i\enskip \psi_{ii}&=&\Id_{\dom \psi_{ii}}, \label{c-rf}\\
\forall ij\enskip\psi_{ij}&=&\psi_{ji}^{-1}, \label{c-sm}\\
\forall ijk\enskip\psi_{ij}\comp\psi_{jk}&\subseteq& \psi_{ik}. \label{c-tr}
\end{eqnarray}
\end{description}

\begin{prop}\label{prop-simeq}
Axiom \ax{AxFn\psi} and \ax{AxCom\psi} imply that
\begin{equation*}
\langle i,\vx\rangle\sim\langle
j,\vy\rangle \defiff \psi_{ij}(\vx)=\vy
\end{equation*}
is an equivalence relation on the set $\{\langle i,\vx\rangle: \vx\in
\dom \psi_{ii}\text{ and }i\in\I\}$, i.e., on the disjoint union of the
domains of $\psi_{ii}$.
\end{prop}
\begin{proof}
Let us first note that $\dom\psi_{ij}\subseteq\dom\psi_{ii}$ for all
$i,j\in\I$ by axiom \ax{AxCom\psi} since
\begin{equation*}
\dom \psi_{ij}=\dom \psi_{ij}\comp
\psi_{ij}^{-1}\stackrel{\eqref{c-sm}}{=}\dom \psi_{ij}\comp
\psi_{ji}\stackrel{\eqref{c-tr}}{\subseteq}\dom\psi_{ii}.
\end{equation*}
Therefore, the definition of $\sim$ is meaningful since we can compute
$\psi_{ij}(\vx)$ for all $\vx\in\dom \psi_{ii}$.

The reflectivity of $\sim$ is equivalent to \eqref{c-rf} since
$\langle i,\vx\rangle\sim\langle i,\vx\rangle$ iff
$\psi_{ii}(\vx)=\vx$ by definition.

The symmetry of $\sim$ is equivalent to \eqref{c-sm} since
\begin{equation*}
\langle i,\vx\rangle\sim\langle j,\vy\rangle \defiff
\psi_{ij}(\vx)=\vy \stackrel{\eqref{c-sm}}{\Iff} \psi_{ji}(\vy)=\vx
\defiff\langle j,\vy\rangle\sim\langle i,\vx\rangle.
\end{equation*}

Finally, the transitivity of $\sim$ is implied by \eqref{c-tr}. To
show this, let $\langle i,\vx\rangle\sim\langle j,\vy\rangle$ and
$\langle j,\vy\rangle\sim\langle k,\vz\rangle$. Then
$\vx\comp\psi_{ij}=\vy$ and $\vy\comp\psi_{jk}=\vz$ by
definition. Hence $ \vx\comp \psi_{ij}\comp\psi_{jk}=\vz$. By
\eqref{c-tr} of axiom \ax{AxCom\psi}, $\psi_{ik}$ extends
$\psi_{ij}\comp\psi_{jk}$. Therefore, $\psi_{ik}(\vx)=\vz$. Thus
$\langle i,\vx\rangle\sim\langle k,\vz\rangle$ as desired.
\end{proof}

\begin{rem}
By Proposition \ref{prop-simeq}, \textbf{manifold} $\Mn$ can be
defined as a new sort in the sense of \cite[p.649.]{logst}, i.e., let
$\Mn$ be the disjoint union of the domains of transition maps
$\psi_{ii}$ factorized by the equivalence relation $\sim$. Let $e\in
\Mn$. The maps 
\begin{equation*}
\psi_i(e)=\vx\defiff \langle i,\vx\rangle \in e
\end{equation*}
are the so called \textbf{charts} of $\Mn$, see
Figure~\ref{fig-lorman}. Chart $\psi_i$ is well-defined since
$\vx=\vy$ if $\langle i,\vx\rangle\sim\langle i,\vy\rangle$
because $\psi_{ii}=\Id_{\dom \psi_{ii}}$.
\end{rem}

\begin{description}
\item[\underline{\ax{AxCDiff\psi}}] The transition maps are
  continuously differentiable:
\begin{multline*}
\forall ij\enskip
\forall \vx\in \dom\psi_{ij}\; \exists \text{ affine map } A \enskip
\psi_{ij}\sim_{\vx} A_{\vx}\text{, and} \\
\forall ij\enskip \forall \varepsilon >0\enskip \exists \delta>0 \enskip \forall
\vy\vz\in\dom\psi_{ij}\;\big( |\vy-\vz|<\delta \to \left|
A_{\vy}-A_{\vz}\right|<\varepsilon\big).
\end{multline*}
\end{description}

Axiom \ax{AxCDiff} implies that $\dom\psi_{ij}$ is open by the definition
of $\sim_{\vx}$, see Remark~\ref{rem-conv}. 

\begin{description}
\item[\underline{\ax{AxCom\g}}] The metric and the transition maps
  commute in the following sense:
\begin{equation*}
\forall i\enskip\forall \vx\in\dom \g_i\cap\dom\psi_{ij}\enskip \forall \vv\vw
 \g_{i}(\vx)(\vv,\vw)=\g_{j}\big(\psi_{ij}(\vx)\big)\big([d_{\vx}\psi_{ij}](\vv),[d_{\vx}\psi_{ij}](\vw)\big).
\end{equation*}
\end{description}

We used Proposition~\ref{prop-simeq} to define the points of manifold
$\Mn$ as equivalence classes of coordinate points connected by the
transition maps $\psi_{ij}$. Proposition~\ref{prop-approx} is an
analogous statement that allows us to tie the vectors of different
coordinate systems into one abstract element of the tangent space at a
certain point $e$ of $\Mn$ by using the derivatives
$[d_{\psi_i(e)}\psi_{ij}]$ of the worldview transformations
$\psi_{ij}$ at the coordinate points $\psi_i(e)$ corresponding to $e$.
\begin{prop} \label{prop-approx}
Let $e\in\Mn$.
Axioms  \ax{AxEField}, \ax{AxFn},  \ax{AxCom\psi} and  \ax{AxCDiff\psi} imply that 
\begin{equation*}
\langle i,\vv\rangle\approx_e\langle
j,\vw\rangle \defiff [d_{\psi_i(e)}\psi_{ij}](\vv)=\vw
\end{equation*}
is an equivalence relation on
the set 
\begin{equation*}
\{\,\langle i,\vv\rangle: \vv\in\Q^d ,\enskip i\in\I\text{ and }
e\in\dom \psi_i\,\}.
\end{equation*}
\end{prop}

\begin{proof}
Axioms \ax{AxEField}, \ax{AxFn}, \ax{AxCom\psi} and \ax{AxCDiff\psi} ensure
that the definition of $\approx_e$ is meaningful, i.e., $\Mn$ is
definable and $[d_{\psi_i(e)}\psi_{ij}]$ exists.

To prove the reflexivity of $\approx_e$, let $i\in\I$ (such that
$e\in\dom\psi_i$) and $\vv\in\Q^d$.  We have  $\langle
i,\vv\rangle\approx_e\langle i,\vv\rangle$ iff
$[d_{\psi_i(e)}\psi_{ii}](\vv)=\vv$. But $\psi_{ii}=\Id_{\psi_{ii}}$ by
\ax{AxCom\psi}. So $[d_{\psi_i(e)}\psi_{ii}]=\Id_{\Q^d}$. Hence
$\approx_e$ is reflexive.

To prove the symmetry of $\approx_e$, let $\langle
i,\vv\rangle\approx_e\langle j,\vw\rangle$. This is equivalent to
$[d_{\psi_i(e)}\psi_{ij}](\vv)=\vw$ by definition. Since
$\psi_{ij}^{-1}=\psi_{ji}$, we have
$[d_{\psi_i(e)}\psi_{ij}]^{-1}=[d_{\psi_j(e)}\psi_{ji}]$ by
Corollary~\ref{cor-invdf}. So $[d_{\psi_j(e)}\psi_{ji}](\vw)=\vv$. Hence
$\langle j,\vw\rangle\approx_e\langle i,\vv\rangle$. Thus $\approx_e$
is symmetric.
 
To prove the transitivity of $\approx_e$, let $\langle
i,\vv\rangle\approx_e\langle j,\vw\rangle$ and $\langle
j,\vw\rangle\approx_e\langle k,\vu\rangle$. Then
$[d_{\psi_i(e)}\psi_{ij}](\vv)=\vw$ and
$[d_{\psi_j(e)}\psi_{jk}](\vw)=\vu$ by definition. By chain rule (see
Theorem~\ref{thm-cr})
$[d_{\psi_i(e)}(\psi_{ij}\comp\psi_{jk})]=[d_{\psi_i(e)}\psi_{ij}]\comp
[d_{\psi_j(e)}\psi_{jk}]$. So $[d_{\psi_i(e)}\psi_{ik}](\vv)=\vu$. Thus
$\langle i,\vv\rangle\approx_e\langle k,\vu\rangle$. Hence $\approx_e$
is transitive.
\end{proof}

\begin{rem}
By Proposition~\ref{prop-approx}, the \textbf{tangent space} at $e\in\Mn$ can
be defined as a new sort in the sense of \cite[p.649.]{logst}, i.e.,
let $\TeM$ be the set 
\begin{equation}
\{\,\langle i,\vv\rangle: \vv\in\Q^d \text{, } i\in\I\text{ and }
e\in\dom\psi_i\,\}
\end{equation}
factorized by the equivalence relation $\approx_e$.  By
\ax{AxCom\g}, the metric $\g$ can be lifted to the tangent space
$\TeM$, i.e., $\gem:\TeM\times\TeM\rightarrow\Q$ can be defined for all
$\bvv,\bvw\in\TeM$ as $\gem(\bvv,\bvw)=\g_i(\vx)(\vv,\vw)$ if
$\psi_i(e)=x$, $\langle i,\vv\rangle\in\bvv$ and $\langle
i,\vw\rangle\in\bvw$. 
\end{rem}

We assume that the metric is Lorentzian by postulating that it can be
transformed to the Minkowski metric $\mu$.

\begin{description}
\item[\underline{\ax{AxLor\g}}] Metric $\g_i$ is a Lorentzian metric for all $i\in\I$:
\begin{equation*}
\forall i\enskip\forall \vx\in\dom\g_i\;\exists \text{ linear map } L \enskip \forall
\vv\vw \enskip \g_{i}(\vx)(\vv,\vw)=\mu(L\vv,L\vw).
\end{equation*}
\end{description}

We also assume that the metric is continuous by the next axiom.
\begin{description}
\item[\underline{\ax{AxC^0\g}}] Metric $g_i$ is continuous for all
  $i$:
\begin{multline*}
\forall i\enskip
\forall \varepsilon>0\enskip
\forall \vx\in\dom\g_i\enskip \exists \delta >0 \enskip \forall
\vy\in\dom\g_i\;\big( |\vx-\vy|<\delta\\ \to
 \forall \vv\vw \; |\g_i(\vx)(\vv,\vw)-\g_i(\vy)(\vv,\vw)|<\varepsilon\big).
\end{multline*}
\end{description}

To ensure that the metric is defined everywhere, we assume the
following axiom.
\begin{description}
\item[\underline{\ax{AxFull\g}}] Metric $\g_i$ is defined everywhere
  on $\dom\psi_{ii}$ for all $i\in\I$:
\begin{equation*}
\forall i\enskip \dom \psi_{ii} \subseteq \dom \g_{i}.
\end{equation*}
\end{description}
Let us note that, without assuming \ax{AxFull\g}, it is even possible that
$\dom \g_i$ is empty for all $i\in\I$.

To be able to introduce our theory \ax{LorMan}, let the above language
of Lorentzian manifolds be denoted by $\LM$, i.e., $\LM\de \{\,\I, \Q,
+,\cdot, \le, \psi,\g\,\}$.

\begin{multline*}
\ax{LorMan}\de \ax{AxEField}+ \ax{AxFn}+ \ax{AxCom\psi}+ \ax{AxCom\g}+
\ax{AxCDiff\psi}\\ +\ax{AxLor\g} +\ax{AxC^0\g} +\ax{AxFull\g}
+\ax{CONT_\LM}
\end{multline*}

\section{Completeness of GenRel with respect to Lorentzian Manifolds}
\label{sec-comp}

Here we are going to define the basic concepts of Lorentzian manifolds
in terms of \ax{GenRel}. This will give us a translation of all the
formulas of the language of Lorentzian manifolds to that of
\ax{GenRel}. Then we will show that the definitional extension of the
models of \ax{GenRel} satisfies the axioms of Lorentzian manifolds,
see Theorem \ref{thm-c00}. This theorem implies that
the translation of any sentence of language $\LM$ of Lorentzian
manifolds can be proved from \ax{GenRel}, see Corollary~\ref{cor-c00}.

Let $d\ge 3$. Let $\G$ be a model of language $\LG$. We are going to
associate a model $M(\G)$ of language $\LM$ to $\G$. Let $M(\G)$ be
the following structure. Let the structure $\langle \Q,+,\cdot,
\le\rangle$ in $M(\G)$ be the same as that of $\G$. Let
\begin{equation*} \I\de \Ob,\enskip \psi_{mk}\de\w_{mk},
\end{equation*}
Finally, let relation $\g$ defined as follows
\begin{equation}\label{eq-g}
\g(m,\vx,\vv,\vw,a) \defiff \exists k \enskip \Ob(k) \land
\W(m,k,\vx)\land \mu\big([d_{\vx}\w_{mk}(\vv)],[d_{\vx}\w_{mk}(\vw)]\big)=a.
\end{equation}

The above model construction determines a translation from language
$\LM$ of Lorentzian manifolds to language $\LG$ of \ax{GenRel}. We
give this translation by formula induction. Quantity variables are
translated to quantity variables and index variables are translated to
body variables. For atomic formulas, it is defined as follows:
\begin{equation*}
Tr(x+y)=x+y,\qquad Tr(x\cdot y)=x\cdot y\qquad Tr(x\le y)=x\le y,
\end{equation*}
\begin{equation*}
Tr\big(\psi(i,j,\vx,\vy)\big)=\forall b \;\big[ \W(i,b,\vx)\leftrightarrow \W(j,b,\vy)\big],
\end{equation*}
\begin{equation*}
Tr\big(\g(i,\vx,\vv,\vw,a)\big)= \exists j\;\big[
\Ob(j)\land\W(i,j,\vx)\land
\mu\big([d_{\vx}\w_{ij}](\vv),[d_{\vx}\w_{ij}](\vv)\big)=a\big].
\end{equation*}
For logical connectives $\land$, $\neg$, etc.:
 \begin{equation*}Tr(\psi\land\varphi)=Tr(\psi)\land Tr(\varphi),\enskip
Tr(\neg\varphi)=\neg Tr(\varphi),\text{ etc.}
\end{equation*}
For quantifiers $\forall$ and $\exists$:
\begin{equation*}Tr(\forall
i \;\phi)=\forall i\;\big[\Ob(i)\rightarrow Tr(\phi)\big] \qquad Tr(\exists i\;
\phi)=\exists i\;\big[ \Ob(i)\land Tr(\phi)\big]
\end{equation*}
if $i$ is an index variable and 
\begin{equation*}Tr(\forall
x \; \phi)=\forall x\; Tr(\phi) \qquad Tr(\exists x\;
\phi)=\exists x\; Tr(\phi)
\end{equation*}
if $x$ is a quantity variable.

As usual, let $\M\models \varphi$ denote that formula $\varphi$ is
valid in model $\M$ and let the class of models of theory $Th$ is defined as
the collection of structures in which all the formulas of $Th$ are
valid:
\begin{equation*}
Mod(Th)\de\{\,\M \::\: \forall \varphi \in Th \enskip \M\models \varphi \,\}.
\end{equation*}

\begin{thm}\label{thm-c00} Let $d\ge3$. Then 
\begin{equation*}
 M(\G)\models\ax{LorMan}\enskip \text{ if }\enskip\G\models\ax{GenRel}
\end{equation*}
or equivalently 
\begin{equation*}
M(\G)\in\Mod(\ax{LorMan})\enskip \text{ if }\enskip \G\in \Mod(\ax{GenRel}),
\end{equation*}
 i.e., $M$ maps models of \ax{GenRel} to models of \ax{LorMan}.
\end{thm}

Theorem~\ref{thm-c00} implies the following completeness of
\ax{GenRel}, where $\vdash$ denotes the usual relation of FOL
deducibility.  Let $\Fm(\Log)$ denote the set of all formulas of
language $\Log$.
\begin{cor}\label{cor-c00} Let $d\ge3$.
Let $\varphi\in \Fm(\LM)$. Then
\begin{equation*}\ax{GenRel}\vdash
Tr(\varphi)\enskip\text{ if }\enskip \M\models \varphi\enskip\text{ for all }\enskip\M\in\Mod(\ax{LorMan}).
\end{equation*}
\end{cor}
The meaning of Corollary~\ref{cor-c00} is that if a statement $\varphi$
is true in every Lorentzian manifold,
then its translation $Tr(\varphi)$ is provable from our axiom system
\ax{GenRel}.

In Section~\ref{sec-generalization}, we generalize these results for
smooth (and $n$-times continuously differentiable) Lorentzian
manifolds, see Theorem~\ref{thm-c0} and Corollary~\ref{cor-c0}.

\section{Turning Lorentzian Manifolds into Models of GenRel}

In Section~\ref{sec-comp}, we have constructed a Lorentzian manifold
from every model of \ax{GenRel}. What about the converse direction?
Can a model of \ax{GenRel} constructed from every Lorentzian manifold?
In this section, we are going to show that the converse construction
works for smooth Lorentzian manifolds if the structure of quantities
is the field $\Reals$ of real numbers.

To outline this construction, let $\M$ be a smooth Lorentzian manifold over
$\Reals$. Let $\langle Q, +, \cdot, \le\rangle$ be $\langle \Reals, +, \cdot,
\le\rangle$. By this choice, \ax{AxEField} and \ax{CONT_\LG} are
satisfied.

A vector $\vv\in\Q^d$ is called \textbf{lightlike} iff the
length of its time component is equal to the length of its space
component, i.e., $|v_t|=|\vv_s|$; or equivalently $\mu(\vv,\vv)=0$.  A
vector $\vv\in\Q^d$ is called \textbf{timelike} iff
$|\vv_s|<|v_t|$; or equivalently iff $\mu(\vv,\vv)>0$. A
differentiable curve is called lightlike (timelike) if all of its
derivative vectors are lightlike (timelike).

Let $\Ph$ be the set of lightlike curves in $\M$.  We associate an
observer $m$ to a normal convex neighborhood\footnote{See, e.g.,
  \cite[pp.129-130]{ONeill} for a precise definition.} $N_m$ and a
timelike curve segment $\gamma_m$ contained by $N_m$. So let $\Ob$ be
the set of pairs consisting a timelike curve segment and normal convex
neighborhood containing it.  Let $\B$ be the union of $\Ph$ and $\Ob$, i.e., $\B=\Ph\cup\Ob$.

To define $W$, we associate a coordinate system to every observer
$m$. Then $W(m,b,\vx)$ will  hold true iff the curve corresponding
to body $b$ crosses coordinate point $\vx$ in observer $m$'s
coordinate system. We define the coordinate system of $m$ as a
transformed version of $N_m$.

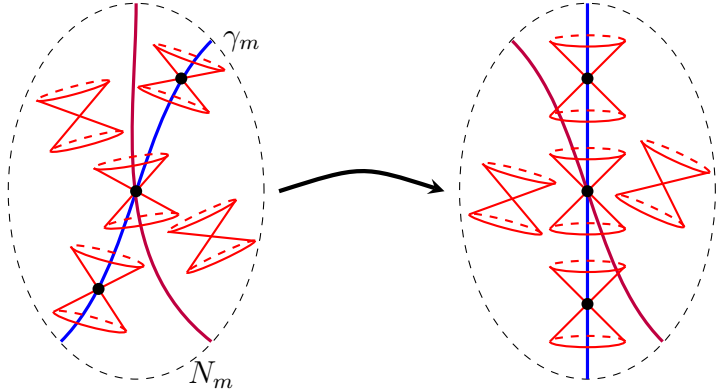
\begin{figure}
\begin{center}
\begin{tikzpicture}[scale=1]
\tikzstyle{nyil}=[->,>=stealth, thick]
\pgfmathsetmacro{\r}{0.08}

\newcommand{\lightcone}[4]{
\begin{scope}[ shift={#3}, scale=#4, rotate=#1, yscale=#2]
\draw[red,thick,dashed] (1,1) arc (0:180:1 and 0.15);
\draw[red,thick,dashed] (1,-1) arc (0:180:1 and 0.15);
\draw[red,thick] (-1,-1)--(1,1);
\draw[red,thick] (1,-1)--(-1,1);
\draw[red,thick] (1,1) arc (0:-180:1 and 0.15);
\draw[red,thick] (1,-1) arc (0:-180:1 and 0.15);
\end{scope}
}

\begin{scope}[shift={(-3,0)}]
\coordinate (o) at (0,0);
\coordinate (a) at (0.6,1.5);
\coordinate (b) at (-0.5,-1.3);
\draw[dashed] (o) ellipse (1.7 and 2.5);
\draw[blue,very thick](-1,-2) .. controls (0,-1) and (0,1) .. (1,2) node[right,black]{$\gamma_m$};
\draw[purple,very thick](1,-2) node[below,yshift=-3,black]{$N_m$}  .. controls (-0.4,-0.9) and (0,1) .. (0,2.5);
\lightcone{-20}{0.6}{(a)}{0.5} 
\lightcone{10}{0.8}{(o)}{0.5} 
\lightcone{-15}{0.9}{(b)}{0.5} 
\lightcone{15}{0.7}{(-0.75,1)}{0.5} 
\lightcone{25}{0.6}{(1,-0.6)}{0.5} 
\fill (o) circle (\r);
\fill (a) circle (\r);
\fill (b) circle (\r);
\end{scope}

\begin{scope}[shift={(3,0)}]
\coordinate (o) at (0,0);
\coordinate (a) at (0,1.5);
\coordinate (b) at (0,-1.5);
\draw[dashed] (0,0) ellipse (1.7 and 2.5);
\draw[blue,very thick](0,-2.5) -- (0,2.5);
\draw[purple,very thick](1,-2) .. controls (0,-1) and (0,1) .. (-1,2);
\lightcone{0}{1}{(o)}{0.5}
\lightcone{0}{1}{(a)}{0.5}
\lightcone{0}{1}{(b)}{0.5}
\lightcone{-15}{0.7}{(-1,-0.1)}{0.5} 
\lightcone{-25}{0.8}{(1,0.1)}{0.5} 
\fill (o) circle (\r);
\fill (a) circle (\r);
\fill (b) circle (\r);
\end{scope}

\draw[nyil,ultra thick] (-1.1,0) .. controls (0,0.35) .. (1.1,0);

\end{tikzpicture}
\caption{\label{fig-FWtransport} Illustration for constructing a model
  of \ax{GenRel} from Lorentzian manifolds }
\end{center}
\end{figure}

To satisfy \ax{AxSelf^-}, we have to transform $N_m$ (in a smooth way)
such that $\gamma_m$ is mapped to a subset of the time-axis, see
Figure~\ref{fig-FWtransport}. To satisfy \ax{AxPh^-}, we have to transform $N_m$
such that the light signals crossing $\gamma_m$ have coordinate speed
1 in the moment of the crossing in $N_m$. These can be ensured by
transforming $N_m$ such that $\gamma_m$ goes to a subset of the time
axis and the metric restricted to the time-axis in transformed $N_m$
is the Minkowski metric. In this case, \ax{AxSymT^-} will also be
satisfied because then the derivative of the worldview transformation
between meeting observers will be a Lorentz transformation in the
point of meeting.


All the axioms corresponding
to the smoothness of coordinate transformations and that of metric are
also satisfied because $\M$ was smooth and we transformed $N_m$
smoothly.

So the only question remains whether it is possible to transform
neighborhoods $N_m$ the way described above? By Fermi--Walker
transporting (see, e.g., \cite[\S 9]{LRR11}) of an orthogonal basis
along $\gamma_m$, we can get a so called Fermi--Walker normal
coordinates. By transforming $N_m$ to this normal coordinates we get
the required coordinate system for observer $m$.

It is a \emph{question for further research} to generalize the
construction of this section for Lorentzian manifolds over Euclidean
fields.\footnote{This question is not at all trivial since it requires
  to generalizing several classical theorems of differential geometry
  over Euclidean fields in the spirit of \cite[\S 10]{Szphd}. This may
  also require extending the languages of \ax{LorMan} and \ax{GenRel},
  e.g., to be able to quantify over integrals of some definable
  functions. This is so because the usual definition of integral (as
  opposed to that of derivative) is not a FOL definition in the
  language of Euclidean fields. In \cite{twp}, we were able to prove
  every theorem over Euclidean fields without a general FOL definable
  concept of integration by quantifying over observers when we needed
  to ensure the existence of the integrals of some definable
  functions. However, this trick may not work to prove every theorem
  used in the construction of this section.}
\section{Geodesics}
\label{sec-geod}

In this section, we are going to define timelike geodesics in
\ax{GenRel}.

We call the worldline of observer $m$ \textbf{timelike geodesic}, if
each of its points has a neighborhood within which $m$ ``maximizes
measured time" between any two encountered events, i.e.,
\begin{multline}\label{eq-tlgeod}
   \forall\vz \in \wl_m(m)\;\exists \delta>0\enskip
   \forall k\vx\vy\; \Big( |\vx-\vz|<\delta\land |\vy-\vz|< \delta
   \land\Ob(k) \\ \land \vx,\vy\in \wl_m(m)\cap\wl_m(k)\land
   \big[\forall \vw\in \wl_m(k) \enskip |\vw-\vz|<\delta \big]\\ \to
   |x_t-y_t| \geq \left|\w_{mk}(\vx)_t-\w_{mk}(\vy)_t\right|\Big),
\end{multline}
see Fig.~\ref{fig-tlgeod}.

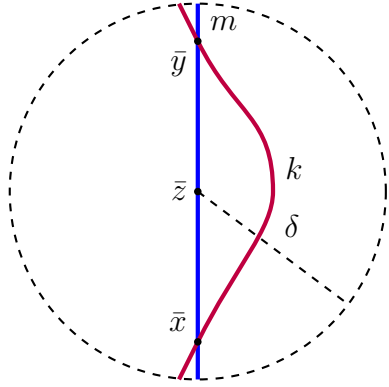
\begin{figure}
\begin{center}
\begin{tikzpicture}[scale=0.5]
\draw[thick, dashed] (0,0) circle (5);
\draw[thick, dashed] (0,0) -- node[above right]{$\delta$} (4,-3);
\draw[ultra thick, blue] (0,-5) -- (0,5) node [below right,black] {$m$};
\draw[ultra thick, purple] (-.5,-5) --(0,-4) .. controls (1,-2) and (2,-1) .. (2,0) node[above right, black] at (2,0) {$k$}  .. controls (2,2) and (1,2) .. (0,4)-- (-0.5,5);
\fill (0,0) circle (.1) node[left]  {$\bar z$};
\fill (0,4) circle (.1) node[below left] {$\bar y$};
\fill (0,-4) circle (.1) node[above left]{$\bar x$};
\end{tikzpicture}
\caption{\label{fig-tlgeod} Illustration for formula \eqref{eq-tlgeod}
  defining timelike geodesics}
\end{center}
\end{figure}

If there are not enough observers, it may not be a big deal that the
worldline of $m$ is a time-like geodesic by the above
definition. Therefore, we postulate the existence of many observers by
the following axiom schema of comprehension.

\begin{description}\label{compr}
\item[\underline{\ax{COMPR}}] For any parametrically definable
  continuously differentiable  timelike curve in any observer's
  worldview, there is another observer whose worldline is the range of
  this curve.
\end{description}
\ax{COMPR} can be formalized as the collection of formulas
\ax{Ax\exists \psi} below. To introduce these formulas, let $\psi$ be
a formula in the language of \ax{GenRel} such that all the free
variables of $\psi$ are among $t$, $\vx$ and $\vy$, where $t\in\Q$,
$\vx\in\Q^d$ and there is no restriction on parameter $\vy$.
\begin{description}
\item[\underline{\ax{Ax\exists \psi}}] If formula $\psi$ defines a
  continuously differentiable timelike curve in observer $m$'s
  worldview, then there is another observer $k$ whose worldline is the
  range of curve $\psi$:
\begin{equation}
\forall \vy \forall m \big(\tlc(m,\psi) \rightarrow \exists k [\vx\in
  \wl_m(k) \leftrightarrow \exists t\enskip \psi(t,\vx,\vy)]\big),
\end{equation}
\end{description}
where $\tlc(m,\psi)$ is a formula expressing that $\psi$ defines a
timelike curve in observer $m$'s worldview.  Formula $\tlc(m,\psi)$ can
be formulated as the conjunction of the following:\\ ``$\psi$ defines
a function,'' i.e.,
\begin{equation*}
\forall t\vx\vz\enskip [\psi(t,\vx,\vy)\land
  \psi(t,\vz,\vy)\rightarrow \vx=\vz],\footnote{From now on, we will
  denote the value of the function defined $\psi$ at $t$ by
  $\psi(t)$.}
\end{equation*}
``$\dom\psi$ is an open interval,'' i.e.,
\begin{equation*}
 \forall ab\in \dom \psi\enskip \exists \delta>0 \enskip \forall c\;
 \big[\big(|a-c|<\delta \lor a<c<b\big)\rightarrow c \in\dom \psi\big],
\end{equation*}
``$\psi$ is differentiable,'' i.e.,
\begin{multline*}
\forall t_0\in\dom\psi\enskip \exists \vz\enskip \forall
\varepsilon>0\enskip \exists \delta>0\enskip \forall
t\;\big(t\in\dom\psi\\ \land  0<|t-t_0|<\delta \rightarrow
|\psi(t)-\psi(t_0)-\vz(t-t_0)|<\varepsilon|t-t_0|\big),\footnotemark
\end{multline*}
\footnotetext{From
  now on, we will denote by $\psi'$ the derivative of $\psi$ defined
  by this FOL formula.}
``$\psi'$ is continuous,'' i.e.,
\begin{equation*}
\forall t_0\in\dom\psi'\enskip \forall
\varepsilon>0\enskip \exists \delta>0\enskip \forall
t\;\big(t\in\dom\psi' \land |t-t_0|<\delta  \rightarrow
|\psi'(t)-\psi'(t_0)|<\varepsilon\big),
\end{equation*}
``$\psi$ is timelike,'' i.e.,
\begin{equation*}
\forall t \in \dom \psi \enskip \forall k\in
\ev_m\big(\psi(t)\big)\;
\mu\big([d_{\psi(t)}\w_{mk}]\psi'(t),[d_{\psi(t)}\w_{mk}]\psi'(t)\big)>0.
\end{equation*}

The assumption of axiom schema \ax{COMPR} guarantees that our
definition of geodesic coincides with the usual one because of the
followings. 

Over the field $\Reals$ of real numbers, a curve is timelike geodesic if it
is locally the longest curve among all timelike curves, see, e.g.,
\cite[Prop.4.5.3.]{Hawking-Ellis}. By \eqref{eq-tlgeod} and
\ax{COMPR}, the worldline of observer is timelike geodesic if it
is locally the longest among all \emph{definable} timelike curves. So
to show that \eqref{eq-tlgeod} gives back the usual notion of timelike
geodesics, it is enough to show that every timelike curve can be
approximated by a \emph{definable} timelike curve. Now we are going to
show this.

\begin{thm}\label{thm-geod}
In continuously differentiable Lorentzian manifolds over the field
$\Reals$ of real numbers, every timelike curve can be approximated
(with arbitrary precision) by continuously differentiable timelike
curves definable in the language of ordered fields.
\end{thm}

\begin{proof}
Let $\gamma$ be a timelike curve that we would like to approximate
with precision $\varepsilon>0$. Without loosing generality we can
assume that $\gamma$ can be covered by one coordinate system
(otherwise we cut $\gamma$ into smaller pieces and approximate it
piece by piece). So let us fix a coordinate system containing
$\gamma$.

It is well-known that all curves can be approximated by broken
lines. So let $\vx_1\vx_2\ldots\vx_n$ be a broken line approximating
$\gamma$ with precision $\varepsilon/3$ in the fixed coordinate system
containing $\gamma$.  Without loosing generality, we can assume that
$\vx_i\vx_{i+1}$ are chords of $\gamma$.

In Minkowski spacetime, all chords of a timelike curve are timelike
(see \cite[Prop.10.4.4]{Szphd} for a proof of this statement using
\ax{AxEField}). So if the broken line approximation is fine enough,
$\vx_i\vx_{i+1}$ are timelike segments since the metric is continuous.

Broken line $\vx_1\vx_2\ldots\vx_n$ may not be definable. However,
since the field $\Rationals$ of rational numbers is dense in $\Reals$
and points having rational coordinates are definable, we can replace
this broken line with a definable one without changing its length more
than $\varepsilon/3$. Let this definable broken line be
$\vy_1\vy_2\ldots\vy_n$. So $\vy_1\vy_2\ldots\vy_n$ is a definable
broken line which approximates $\gamma$ with precision
$2\varepsilon/3$.

We have to prove that the vertexes of this definable broken line
$\vy_1\vy_2\ldots\vy_n$ can be rounded by continuously differentiable
definable timelike curves in small enough neighborhoods without
changing its length more than $\varepsilon/3$.  Since the vertexes of
$\vy_1\vy_2\ldots\vy_n$ are points having rational coordinates, the
corresponding four-velocities ($\vy_i-\vy_{i-1}$) are definable and
the definable coordinates on $\vy_1\vy_2\ldots\vy_n$ are dense. In
Minkowski spacetime, any two definable coordinate points can be
connected by a continuously differentiable definable timelike curve
$\gamma^*$ such that the speeds of $\gamma^*$ at the start and the end
are arbitrary definable speeds smaller than $1$ and the speed of
$\gamma^*$ between these points are smaller than $1-\delta$ for some
$\delta>0$ (this last property guaranties that $\gamma^*$ remains
timelike if we change the metric slightly), see
Lemma~\ref{lem-deftlc}.  Since the metric is continuous, the metric in
small enough neighborhoods around the vertexes are approximately the
Minkowski metric. So we can use Lemma~\ref{lem-deftlc} to connect
definable points on the edges near to the vertexes of the broken line
$\vy_1\vy_2\ldots \vy_n$ in small enough neighborhoods without
changing its length more than $\varepsilon/3$.

The resulting rounded up broken line can be parametrized such that it
gives us the desired continuously differentiable definable timelike
curve that approximates $\gamma$ with precision $\varepsilon$.
\end{proof}

\section{refinements of the main theorem}
\label{sec-generalization}

In this section, we are going to refine Theorem~\ref{thm-c00} for
smooth (and $n$-times continuously differentiable) Lorentzian
manifolds by introducing axioms ensuring the smoothness of the
worldview transformations and the metric. To do so, we need some
further definitions.

Let the \textbf{standard basis vectors} of $\Q^n$ be denoted by $\ve_i$, i.e.,
\begin{equation*}
\ve_i\de \langle0,\ldots,\stackrel{i}{1},\ldots,0\rangle
\end{equation*}
for all $1\le i\le n$.  Let $f$ be a definable function from a subset
of $\Q^k$ to $Q$ defined by formula $\phi_f(\vx,y)$, i.e.,
$f(\vx)=y\Iff\phi_f(\vx,y)$.  The $i$-ht \textbf{partial derivative}
of $f$ is defined by the following FOL formula:
\begin{multline*}
\phi_{\partial_if}(\vz,w)\defiff \phi_f(\vx,y)\land\forall
\varepsilon>0\enskip \exists \delta >0\enskip \forall
hz\enskip\\ \big(|h|\le\delta \land \phi_f(\vx+h\cdot\ve_i,z) \rightarrow
|z-y-w\cdot h|\le \varepsilon |h|\big).
\end{multline*}

\noindent
Formula $\phi_{\partial_if}$ captures the usual concept of partial
derivatives, i.e.,
\begin{equation*}
\partial_if(\vx)\de\lim_{h\rightarrow
  0}\frac{f(\vx+h\cdot\ve_i)-f(\vx)}{h}.
\end{equation*}

We say that $i$-th partial derivative of $f$ exists at $\vz$ iff there
is a $w$ such that $\phi_{\partial_if}(\vz,w)$ holds. Since function
$\partial_if$ defined by formula $\phi_{\partial_if}$ is the same type
as $f$, i.e., $\dom \partial_if\subseteq\dom f$ and $\ran
\partial_if\subseteq \Q$, we can iterate the partial derivations and
define $\partial_{i_1\ldots i_n}f$ as
$\partial_{i_1}\partial_{i_{2}}\ldots\partial_{i_n}f$.

Function $f=\langle f_1,\ldots, f_m\rangle: \Q^k\rightarrow \Q^m$ is
said to be $n$-times \textbf{continuously differentiable} if $\dom f$
is open and all $n$-th partial derivatives of all of its components
(i.e., $\partial_{i_1\ldots i_n}f_1(\vz),\ldots,\partial_{i_1\ldots
  i_n}f_m(\vz)$ for all $0\le i_1\ldots i_n\le k$) exits for all
$\vz\in\dom f$ and they are continuous. This concept can be defined by
a FOL formula since the partial derivatives and the continuity can be
defined in FOL, see \cite[\S 10.2]{Szphd}.

In \ax{GenRel}, we assumed only the differentiability of the worldview
transformations.  We assume stronger differentiability properties for
them by the next axioms:

\begin{description}
\item[\underline{\ax{AxC^n}}] The worldview transformations are
  $n$-times continuously differentiable maps:
\begin{multline*}
\forall m,k\in\Ob\;\big(\w_{mk} \text{ is a function } \\\land  \dom\w_{mk} \text{ is open \footnotemark} \land \forall
\vx\in \dom\w_{mk} \\\bigwedge_{1\le a_1,\ldots,a_n\le d}
\partial_{a_1\ldots a_n}\w_{mk}(\vx) \text{ exists } \text{and }
\partial_{a_1\ldots a_n}\w_{mk}\text{ is continuous.\big)\footnotemark}
\end{multline*}
\addtocounter{footnote}{-1}
\footnotetext{The statement ``definable set $H\subseteq \Q^n$ is open'' can be
  captured by the FOL formula $\forall \vx\in H \;
  \exists \delta >0 \;\forall \vy\;(|\vx-\vy|<\delta \rightarrow \vy\in
  H$).}  
\addtocounter{footnote}{1}\footnotetext{The continuity of definable function $f$ can be captured by
  the following FOL formula $\forall \vx\in\dom f \;
  \forall \varepsilon >0 \; \exists \delta >0 \; \forall \vy \in \dom
  f\;(|\vx-\vy|<\delta \rightarrow |f(\vx)-f(\vy)|<\varepsilon$).} 
\end{description}
\begin{rem}\label{rem-cdiff} Axiom \ax{AxCDiff} is equivalent to \ax{AxC^1} because, if $f$ is a differentiable function from a subset of $\Q^k$ to $\Q^m$, then
\begin{equation*}
[d_{\vx}f]=\begin{bmatrix} \partial_1f_1(\vx)
&\partial_2f_1(\vx)&\ldots&\partial_kf_1(\vx)\\ \partial_1f_2(\vx)
&\partial_2f_2(\vx)&\ldots&\partial_kf_2(\vx)\\ \vdots&\vdots&\ddots&\vdots\\ \partial_1f_m(\vx)
&\partial_2f_m(\vx)&\ldots&\partial_kf_m(\vx)
\end{bmatrix}.
\end{equation*}
\end{rem}

By the following axioms, we can ensure the metric corresponding
observers to be smooth enough.
\begin{description}
\item[\underline{\ax{AxC^n\g_m}}] For all observer $m$, the metric
  $g_m$ is $n$-times continuously differentiable:
\begin{equation*}
\forall m\in\Ob\enskip
\forall \vx\in\dom\g_m\enskip \bigwedge_{1\le a_1,\ldots,a_n\le d}
\partial_{a_1\ldots a_n} \g_m(\vx) \text{ exists and }
\partial_{a_1\ldots a_n} \g_m \text{ is continuous}.
\end{equation*}
\end{description}

For the smooth case, let us introduce \ax{AxC^\infty} as the axiom
schema containing \ax{AxC^n} for all positive integers $n$; and let
\ax{AxC^\infty \g_m} be the axiom schema containing \ax{AxC^n\g_m} for
all positive integers $n$.  Now we can introduce the promised
extensions of \ax{GenRel}:
\begin{equation*}
\ax{GenRel^{n}}\de\ax{GenRel}+\ax{AxC^n}+\ax{AxC^{n-1}\g_m}
\end{equation*}
if $1\le n\le \infty$.  By Remark~\ref{rem-cdiff}, \ax{GenRel^1} is equivalent to \ax{GenRel}.

Let us now introduce the corresponding axioms for $n$-times
continuously differentiable Lorentzian manifolds.
\begin{description}
\item[\underline{\ax{AxC^n\psi}}] The transition maps are $n$-times continuously
  differentiable:
\begin{multline*}
\forall ij \enskip 
\dom\psi_{ij} \text{ is open } \land\forall \vx\in \dom\psi_{ij}
\\\bigwedge_{1\le a_1,\ldots,a_n\le d} \partial_{a_1\ldots
  a_n}\psi_{ij}(\vx) \text{ exists } \text{and } \partial_{a_1\ldots
  a_n}\psi_{ij}\text{ is continuous.}
\end{multline*}
\end{description}

\begin{description}
\item[\underline{\ax{AxC^n\g}}] Metric $g_i$ is $n$-times continuously
  differentiable for all $i$:
  \begin{equation*}
\forall i\enskip
 \forall \vx\in\dom\g_i\enskip
 \bigwedge_{1\le a_1,\ldots,a_n\le d} \partial_{a_1\ldots a_n}
\g_i(\vx) \text{ exists and } \partial_{a_1\ldots a_n} \g_i \text{ is
  continuous}.
\end{equation*}
\end{description}

Let \ax{AxC^\infty\psi} be the axiom schema containing \ax{AxC^n\psi}
for all positive integers $n$; and let \ax{AxC^\infty \g} be the axiom
schema containing \ax{AxC^n\g} for all positive integers $n$.

Now we can introduce the FOL theories for Lorentzian manifolds
corresponding to \ax{GenRel^n}:
\begin{equation*}
\ax{LorMan^{n}}\de \ax{LorMan} + \ax{AxC^{n}\psi}+\ax{AxC^{n-1}\g} 
\end{equation*}
if $1\le n\le \infty$.

\begin{thm}\label{thm-c0} Let $d\ge3$ and $1\le n\le\infty$. Then 
\begin{equation*}
 M(\G)\models\ax{LorMan^{n}}\enskip \text{ if }\enskip\G\models\ax{GenRel^n}
\end{equation*}
or equivalently 
\begin{equation*}
M(\G)\in\Mod(\ax{LorMan^{n}})\enskip \text{ if }\enskip \G\in \Mod(\ax{GenRel^n}),
\end{equation*}
 i.e., $M$ maps models of \ax{GenRel^n} to models of \ax{LorMan^n}.
\end{thm}

Theorem~\ref{thm-c0} implies the following completeness of
\ax{GenRel}, where $\vdash$ denotes the usual relation of FOL
deducibility.
\begin{cor}\label{cor-c0} Let $d\ge3$ and $1\le n\le\infty$.
Let $\varphi\in \Fm(\LM)$. Then
\begin{equation*}\ax{GenRel^{n}}\vdash
Tr(\varphi)\enskip\text{ if }\enskip \M\models \varphi\enskip\text{ for all }\enskip\M\in\Mod(\ax{LorMan^{n}}).
\end{equation*}
\end{cor}

\section{Proof of Theorem~\ref{thm-c0}}\label{sec-proof}

In this section, we are going to prove our main result
Theorem~\ref{thm-c0} and some earlier used statements. To do so, let
us first give a detailed introduction of axiom schema \ax{CONT_\Log}.

Let $\Log$ be a many sorted language containing sort $\Q$ and a binary
relation $\le$ on $\Q$. Let $\Fm(\Log)$ be the set of FOL formulas of $\Log$.

\label{p-detailedcont}
To introduce a \ax{CONT_\Log} precisely, we have to introduce some
notations.  Let $\M$ be a model of language $\Log$ and
$\varphi\in\Fm(\Log)$. Let $U$ be the union of the sorts of $\M$.  We
use $\M\models \varphi$ in the usual sense of mathematical logic to
denote that formula $\varphi$ is valid in the structure $\M$ and
$\M\models \varphi[a_1,\ldots,a_n]$ to denote that $a_1,\ldots,a_n\in
U$ satisfies $\varphi$ in $\mathfrak{M}$.  We say that a subset $H$ of
$\Q$ is \textbf{(parametrically) $\Log$-definable by} $\varphi$ iff
there are $a_1,\ldots,a_n\in U$ such that
\begin{equation*}
H=\Setopen d\in\Q\setmid\mathfrak{M}\models\varphi[d,a_1,\ldots,a_n]
\Setclose. 
\end{equation*}
We say that a subset of $\Q$ is \textbf{$\Log$-definable} iff it is
definable by a formula of $\Log$. More generally, an $n$-ary relation
$R\subseteq \Q^n$ is said to be \textbf{$\Log$-definable} in
$\mathfrak{M}$ by parameters iff there is a formula
$\varphi\in\Fm(\Log)$ with only free variables $x_1,\ldots,x_n,
y_1,\ldots,y_{k}$ and there are $a_{1},\ldots,a_{k}\in U$ such that
\begin{equation*}
R=\Setopen\langle p_1,\ldots, p_n\rangle\in \Q^n : \mathfrak{M}\models \varphi[p_1,\ldots,p_n,a_1,\ldots,a_k]\Setclose.
\end{equation*}
By the next axiom, for all formulas $\varphi\in\Fm(\Log)$ defining a
subset of the quantities, we introduce an axiom postulating the
existence of the supremum of the defined set if it is not empty and
bounded.
\begin{description}
\item[\ax{AxSup^{\Log}_\varphi}] 
Every subset of $\Q$ definable by $\varphi$ (when using $y_1,\ldots,y_n$ as fixed parameters) has a supremum if it is nonempty and bounded:
\begin{equation*}
\forall y_1,\ldots,y_n\enskip\big(\exists x\enskip\varphi\big)\land \big[\big(\exists b\enskip \forall x\enskip\varphi\rightarrow x\le b\big) \rightarrow
\big(\exists s\enskip \forall b\enskip [\forall x\enskip
\varphi\rightarrow x\le b] \leftrightarrow s\le b\big)\big],
\end{equation*}
\end{description}
where $x$ is a variable of sort $\Q$.
Now we can introduce \ax{CONT_\Log} at the following axiom schema:
\begin{equation*} 
\ax{CONT_\Log} \de \Setopen \ax{AxSup_\varphi}\setmid
\varphi \text{ is a FOL formula of language }\Log \Setclose.
\end{equation*}
Let us note that \ax{CONT_\Log} is true in any model whose structure of
quantities is the field of real numbers.

Let us also recall here the definition of Lorentz transformation.  A
linear transformation $L$ is called \textbf{Lorentz transformation}
iff it preserves the Minkowski metric $\mu$, i.e.,
$\mu(\vv,\vw)=\mu\big(L(\vv),L(\vw)\big)$ for all $\vv,\vw\in\Q^d$.

Theorem~\ref{thm-lordf} states that \ax{GenRel} implies that the
derivatives of the worldview transformations between observers at the
events of meeting are Lorentz transformations.

\begin{thm}\label{thm-lordf} Let $d\ge3$. Assume
\ax{GenRel}. Let $ \forall m,k\in\Ob$ and
$\vx\in\wl_m(k)\cap\wl_m(m)$. Then $\w_{mk}$ is differentiable at
$\vx$ and $[d_{\vx}\w_{mk}]$ is a Lorentz transformation.
\end{thm}

Here we are going to prove Theorem~\ref{thm-lordf}.  
To do so, first we introduce
some definitions and lemmas we will use in the proof.

\begin{lemma}\label{lem-df0}
Assume \ax{AxEField}, \ax{AxEv^-}, and \ax{AxCDiff}. Let $m,k\in\Ob$ and
$\vx\in\wl_m(k)$. Then $\w_{mk}$ is a function differentiable at
$\vx$.
\end{lemma}
\begin{proof} Since $\vx\in \wl_m(k)$, there is a $\vy$ such
that $\ev_m(\vx)=\ev_k(\vy)$ by \ax{AxEv^-}. We have that
$\ev_m(\vx)\neq\emptyset$ since $k\in\ev_m(\vx)$. Hence
$\vx\in\dom\w_{mk}$. So by \ax{AxCDiff},
$\w_{mk}$ is a function differentiable at $\vx$.
\end{proof}

Let us recall here that the chain rule of real analysis can be proved
using axiom \ax{AxEField} only, see \cite[\S 10.3]{Szphd}.

\begin{thm}[chain rule]\label{thm-cr}
Assume \ax{AxEField}.  Let $g:\Q^n \rightarrow \Q^m$ and
$f:\Q^m\rightarrow \Q^k$.  If $g$ is differentiable at $\vx\in \Q^n$
and $f$ is differentiable at $g(\vx)$, then $g\comp f$ is differentiable at
$\vx$ and its derivative is $[d_{\vx}g]\comp [d_{g(\vx)}f]$, i.e.,
\begin{equation*} 
[d_{\vx}(g \comp f)] = [d_{\vx}g] \comp [d_{g(\vx)}f].
\end{equation*}
In particular, if $g: \Q\rightarrow \Q^m$, and $g$ is differentiable
at $x\in\Q$ and $f$ is differentiable at $g(x)$, then
\begin{equation*} 
(g\comp f)'(x)=[d_{g(x)}f]\big(g'(x)\big).
\end{equation*} 
\end{thm}

\begin{cor}\label{cor-invdf}
Assume \ax{AxEField}. Let $f:\Q^n\rightarrow\Q^n$ be an injective
function such that $f^{-1}$ is differentiable at $\vx$ and $f$ is
differentiable at $f^{-1}(\vx)$. Then
\begin{equation*}
[d_{\vx}f^{-1}]=[d_{f^{-1}(\vx)}f]^{-1}.
\end{equation*}
In particular, if $n=1$,
\begin{equation*}
(f^{-1})'(x)=\frac{1}{f'\big(f^{-1}(x)\big)}.
\end{equation*}
\end{cor}

\begin{lemma}\label{lem-df}
Assume \ax{AxEField}, \ax{AxEv^-}, and \ax{AxCDiff}. Let $m,k\in\Ob$ and
$\vx\in\wl_m(k)\cap\wl_m(m)$.  Then $[d_{\vx}\w_{mk}]$ is invertible and $[d_{\vx}\w_{mk}]^{-1}=[d_{\vy}\w_{km}]$, where  $\vy=\w_{mk}(\vx)$.
\end{lemma}
\begin{proof} By \ax{AxCDiff}, $\w_{mk}$ and $\w_{km}$ are 
differentiable functions. Since $\vx\in \wl_m(k)$, there is a $\vy$
such that $\ev_m(\vx)=\ev_k(\vy)$ by \ax{AxEv^-}. We have that
$\ev_m(\vx)\neq\emptyset$ since $m,k\in\ev_m(\vx)$. Hence
$\vx\in\dom\w_{mk}$, $\vy=\w_{mk}(\vx)$ and $\vy\in\dom\w_{km}$. Thus
$\w_{mk}$ is differentiable at $\vx$ and $\w_{km}$ is differentiable
at $\vy$. Since $\w_{km}$ is the inverse of $\w_{mk}$ by definition,
 $[d_{\vx}\w_{mk}]$ is invertible and its inverse is
$[d_{\vy}\w_{km}]$ by Corollary \ref{cor-invdf}.
\end{proof}

The \textbf{restriction} of function $f:A\rightarrow B$ to set $H$,
denoted by $f\restr_H$, is defined as follows:
\begin{equation*}
f\restr_H\,\de\{\, \langle a,b\rangle : a\in H\cap \dom f \land
f(a)=b\,\}.
\end{equation*}
The \textbf{$f$-image of set $H$}, is defined as follows:
\begin{equation*}
f[H]=\{\,b:\exists a\in H\cap \dom f \land f(a)=b\,\}.
\end{equation*}

\begin{lemma}\label{lem-ll}
Assume \ax{AxEField}, \ax{AxEv^-}, \ax{AxCDiff}, and \ax{AxPh^-}. Let
$m,k\in\Ob$ and $\vx\in\wl_m(k)\cap\wl_m(m)$. Then  $[d_{\vx}\w_{mk}]$ is a linear
bijection taking lightlike vectors to lightlike vectors.
\end{lemma}

\begin{proof} 
By Lemma~\ref{lem-df}, $[d_{\vx}\w_{mk}]$ is a linear bijection.

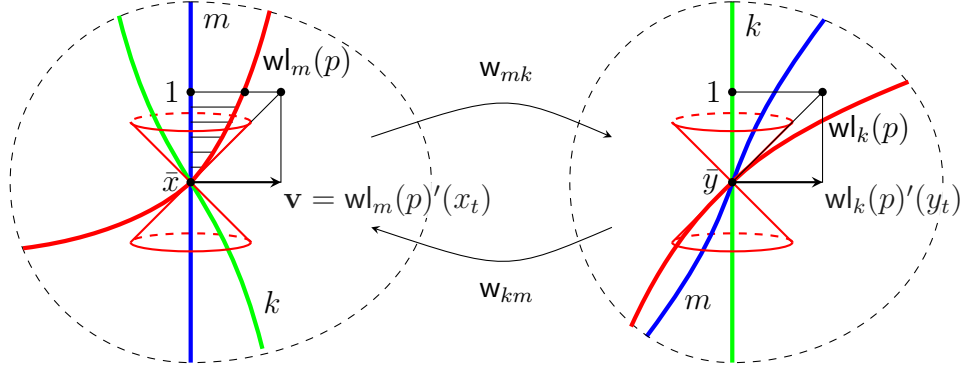
\begin{figure}
\begin{center}
\begin{tikzpicture}[scale=0.8]
\tikzset{>=stealth}
\tikzstyle{cimke}=[fill opacity=0.9, inner sep=0.5, fill=white]
\def\r{0.07}

\begin{scope}[shift={(-4,0)}]
\draw[red,thick,dashed] (1,1) arc (0:180:1 and 0.15);
\draw[red,thick,dashed] (1,-1) arc (0:180:1 and 0.15);
\draw[ultra thick, blue] (0,-3) -- (0,3) node[below right,black] {$m$};
\draw[ultra thick, green, shorten >=-18, shorten <=-18] (-1,2) .. controls (-0.5,0) and (0.5,0)..  (1,-2)node[right,black] {$k$};
\draw[thin, dashed] (0,3)  arc (90:270:3 and 3);
\draw[thin, dashed] (0,3) arc (90:-90:4 and 3);
\draw (0,1.5)--(1.5,1.5) -- (0,0) (1.5,1.5) --(1.5,0)--(0,0) ;
\draw[<-,thick] (1.5,0) node[cimke, below right]{$\mathbf{v}=\mathsf{wl}_m(p)'(x_t)$}--(0,0);
\draw (0,1.25)--(0.8,1.25);
\draw (0,1)--(0.7,1);
\draw (0,0.75)--(0.5,.75);
\draw (0,0.5)--(0.4,.5);
\draw (0,0.25)-- (0.2,.25);
\draw[red, ultra thick, shorten >=-20, shorten <=-18] (-2,-1)  .. controls (-1.2,-0.9) and (-.3,-.3)  .. (0,0) .. controls (.3,.3) and (0.9,1.2) ..(1.1,2)node[right,black,xshift=-2.5]{$\mathsf{wl}_m(p)$};
\draw[red,thick] (-1,-1)--(1,1);
\draw[red,thick] (1,-1)--(-1,1);
\draw[red,thick] (1,1) arc (0:-180:1 and 0.15);
\draw[red,thick] (1,-1) arc (0:-180:1 and 0.15);
\fill (0,0) node[left] {$\bar x$} circle (\r);
\fill (0,1.5) node[left] {$1$} circle (\r);
\fill (0.9,1.5) circle (\r);
\fill (1.5,1.5) circle (\r);
\end{scope}

\begin{scope}[shift={(5,0)}]
\draw[red,thick,dashed] (1,1) arc (0:180:1 and 0.15);
\draw[red,thick,dashed] (1,-1) arc (0:180:1 and 0.15);
\draw[ultra thick, green] (0,-3) -- (0,3)  node[below right,black] {$k$};
\draw[ultra thick, blue, shorten >=-16, shorten <=-20] (1,2) .. controls (-0.5,0) and (0.5,0)..  (-1,-2)node[right,black] {$m$};
\draw[red,thick] (-1,-1)--(1,1);
\draw[red,thick] (1,-1)--(-1,1);
\draw[red,thick] (1,1) arc (0:-180:1 and 0.15);
\draw[red,thick] (1,-1) arc (0:-180:1 and 0.15);
\draw[thin, dashed] (0,3)  arc (90:270:2.7 and 3);
\draw[thin, dashed] (0,3) arc (90:-90:3.5 and 3);
\draw[red, ultra thick, shorten >=-35, shorten <=-10] (-1.5,-2)  .. controls (-1,-0.9) and (-.3,-.3)  .. (0,0) .. controls (.3,.3) and (1,0.9) ..(1.5,1.1)node[right,yshift=-5,black,xshift=-2.5]{$\mathsf{wl}_k(p)$};
\fill (0,0) node[left,xshift=-1]{$\bar y$} circle (\r);
\draw (0,1.5)--(1.5,1.5) -- (0,0) (1.5,1.5) --(1.5,0)--(0,0) ;
\draw[<-,thick] (1.5,0) node[cimke, below right]{$\mathsf{wl}_k(p)'(y_t)$}--(0,0);
\fill (0,1.5) node[left] {$1$} circle (\r);
\fill (1.5,1.5) circle (\r);
\end{scope}

\coordinate (a) at (1.25,1.5);
\node[above] at (a) {$\mathsf{w}_{mk}$};
\draw[->] (-1,0.75) .. controls (a).. (3,0.75);
\coordinate (b) at (1.25,-1.5);
\node[below] at (b) {$\mathsf{w}_{km}$};
\draw[<-] (-1,-0.75) .. controls (b).. (3,-0.75);)

\end{tikzpicture}
\caption{\label{fig-phtr} Illustration for the proof of Lemma~\ref{lem-ll}}
\end{center}
\end{figure}

Now we are going to show that $[d_{\vx}\w_{mk}]$ takes lightlike vectors
to lightlike ones. To do so, let $\bv\in\Q^{d-1}$ for which $|\bv|=1$.
Since $\vx\in \wl_m(m)$, there is a photon $p$ in event $\ev_m(\vx)$
such that $\bv=\wl_m(p)'(x_t)$ by \ax{AxPh^-}, see
Figure~\ref{fig-phtr}.  Let $\vy$ be the $\w_{mk}$ image of $\vx$. By
\ax{AxPh^-}, $\wl_k(p)$ is a function defined in an open neighborhood
of $y_t$.  Since $\dom \w_{mk}$ and $\ran \w_{mk}$ are open, and
$\w_{mk}$ is continuous, there is an open set
$U\subseteq Q^d$ such that $\vx\in U$ and
\begin{equation*}\ran
\w_{mk}[\wl_m(p)\cap U]\subseteq \wl_k(p).
\end{equation*}
Therefore, the tangent line of $\wl_m(p)$ is mapped into the tangent
line of $\wl_k(p)$ by $[d_{\vx}\w_{mk}]$.  Thus $[d_{\vx}\w_{mk}](\langle
\bv,1\rangle)$ is parallel to $\langle \wl_k(p)'(y_t),1\rangle$, which
is a lightlike vector since $|\wl_k(p)'(y_t)|=1$ by \ax{AxPh^-}.
Therefore, $[d_{\vx}\w_{mk}](\langle \bv,1\rangle)$ is a lightlike
vector. Since for any lightlike vector $\vv\in\Q^d$ there is a
 $\bv\in\Q^{d-1}$ and $c\in\Q$ such that
$\vv=c\cdot\langle\bv,1\rangle$, we have that $[d_{\vx}\w_{mk}]$ is a
linear transformation taking lightlike vectors to lightlike vectors.
\end{proof}

We say that a linear bijection $A$ has the
\textbf{sym-time property} if 
\begin{equation*}
A(\ve_d)_t=A^{-1}(\ve_d)_t.
\end{equation*}

\begin{lemma}\label{lem-sym} 
Assume axioms \ax{AxEField}, \ax{AxSelf^-}, \ax{AxEv^-}, and
\ax{AxCDiff}. Then \ax{AxSymt^-} implies that $[d_{\vx} \w_{mk}]$ has the
sym-time property for all observers $m$ and $k$ and coordinate point
$\vx$ for which $m,k\in\ev_m(\vx)$.
\end{lemma}
\begin{proof}

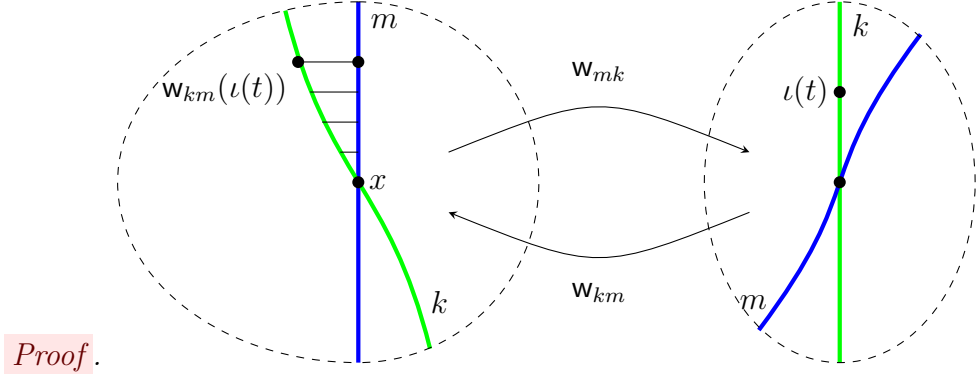
\begin{figure}
\begin{center}
\begin{tikzpicture}[scale=0.8]
\tikzset{>=stealth}
\def\r{0.1}

\begin{scope}[shift={(-4,0)}]
\draw[ultra thick, blue] (0,-3) -- (0,3) node[below right,black] {$m$};
\draw[ultra thick, green, shorten >=-18, shorten <=-20] (-1,2) .. controls (-0.5,0) and (0.5,0)..  (1,-2)node[right,black] {$k$};
\draw[thin, dashed] (0,3)  arc (90:270:4 and 3);
\draw[thin, dashed] (0,3) arc (90:-90:3 and 3);
\draw (-1,2) node[below left] {$\mathsf{w}_{km}(\iota(t))$}--(0,2);
\draw (-0.3,0.5) --(0,0.5);
\draw (-0.6,1) -- (0,1);
\draw (-0.8,1.5) --(0,1.5);
\fill (0,0) node[right] {$x$} circle (\r);
\fill (-1,2) circle (\r);
\fill (0,2) circle (\r);
\end{scope}

\begin{scope}[shift={(4,0)}]
\draw[ultra thick, green] (0,-3) -- (0,3)  node[below right,black] {$k$};
\draw[ultra thick, blue, shorten >=-13, shorten <=-13] (1,2) .. controls (-0.5,0) and (0.5,0)..  (-1,-2)node[left,black] {$m$};
\draw[thin, dashed] (0,0) ellipse (2.25 and 3);
\fill (0,0) circle (\r);
 \fill (0,1.5) node[left] {$\iota(t)$} circle (\r);
\end{scope}

\coordinate (a) at (0,1.5);
\node[above] at (a) {$\mathsf{w}_{mk}$};
\draw[->] (-2.5,0.5) .. controls (a).. (2.5,0.5);
\coordinate (b) at (0,-1.5);
\node[below] at (b) {$\mathsf{w}_{km}$};
\draw[<-] (-2.5,-0.5) .. controls (b).. (2.5,-0.5);)

\end{tikzpicture}
\caption{\label{fig-cltr}Illustration for the proof of Lemma~\ref{lem-sym} }
\end{center}
\end{figure}

Let $m$ and $k$ be observers and let $\vx$ be a coordinate point such
that $m,k\in\ev_m(\vx)$.  By \ax{AxCDiff}, $\w_{mk}$ is a
differentiable function. By Lemma~\ref{lem-df0}, $\w_{mk}$ is
differentiable at $\vx$, i.e., $\vx\in\dom\w_{mk}$.  Let $\vy$ be
$\w_{mk}(\vx)$.  Let $\iota:\Q\rightarrow\Q^d$ be the linear map
$\iota(t)\de\langle 0,\ldots,0,t\rangle$ for all $t\in\Q$ and let the
projection $\pi_t:\Q^d\rightarrow\Q$ be defined as $\pi_t(\vx)\de x_t$
for all $\vx\in\Q^d$.  By axiom \ax{AxSelf^-}, $\vy=\iota(y_t)$ since
$\W(k,k,\vy)$. Let us note that
\begin{equation*}
\cl_{mk}(t)=(\iota\comp \w_{km}\comp \pi_t)^{-1}(t)
\end{equation*}
for all $t\in\dom\cl_{mk}$ by definitions and axiom \ax{AxSelf^-}, see
Figure~\ref{fig-cltr}. So $\cl_{mk}(x_t)=y_t$ since $\iota(y_t)=\vy$,
$\w_{km}(\vy)=\vx$ and $\pi_t(\vx)=x_t$.  Thus, by
Corollary~\ref{cor-invdf},
\begin{equation*}\cl_{mk}'(x_t)=\frac{1}{(\iota\comp \w_{km}\comp \pi_t)'(y_t)}
\end{equation*}
since $\cl_{mk}$ is differentiable at $x_t$ by \ax{AxSymt^{-}} and its
inverse $\iota\comp\w_{km}\comp\pi_t$ is a differentiable map by
\ax{AxCDiff} and the fact that $\pi_t$ and $\iota$ are linear maps.
Since $\iota$ and $\pi_t$ are linear maps, $[d_{\vz}\pi_t]=\pi_t$ for
all $\vz\in\Q^d$ and $\iota'(t)=[d_{t}\iota](1)=\iota(1)=\ve_d$ for all
$t\in\Q$. Thus, by chain rule (Theorem~\ref{thm-cr}), we have
\begin{equation*} 
(\iota \comp\w_{km}\comp \pi_t)'(y_t)=\big( [d_{\iota(y_t)}\w_{km}]\comp
  \pi_t\big)\big(\iota'(y_t)\big)=\big([d_{\vy}\w_{km}](\ve_d)\big)_t.
\end{equation*} Therefore, 
\begin{equation*}
\cl_{mk}'(x_t)=\frac{1}{\big([d_{\vy}\w_{km}](\ve_d)\big)_t}.
\end{equation*}
Similarly,
\begin{equation*}
\cl_{km}'(y_t)=\frac{1}{\big([d_{\vx}\w_{mk}](\ve_d)\big)_t}.
\end{equation*}
 By \ax{AxSymt^-}, we
have $\cl_{mk}'(x_t)=\cl_{km}'(y_t)$. 
Consequently,  
\begin{equation*}
\big([d_{\vy}\w_{km}](\ve_d)\big)_t=\big([d_{\vx}\w_{mk}](\ve_d)\big)_t.
\end{equation*}
By Lemma~\ref{lem-df}, $[d_{\vx}\w_{mk}]^{-1}=[d_{\vy}\w_{km}]$.  Thus
\begin{equation}
\big([d_{\vx}\w_{mk}](\ve_d)\big)_t=\big([d_{\vx}\w_{mk}]^{-1}(\ve_d)\big)_t.
\end{equation}
Therefore,  $[d_{\vx}\w_{mk}]$ has the sym-time property; and this is what we
wanted to prove.
\end{proof}

We call a linear bijection of $\Q^d$ \textbf{space isometry} iff it is an
isometry on the space part of $\Q^d$ fixing $\ve_d$, i.e.,
$M(\ve_d)=\ve_d$, $|M(\vx)|=|\vx|$ and $M(\vx)_t=0$ for all
  $\vx\in\Q^d$ for which $x_t=0$.

\begin{lemma}\label{lem-a} 
Assume \ax{AxEField}. Any linear bijection $M$ taking lightlike vectors to
lightlike ones fixing $\ve_d$ is a space isometry.
\end{lemma}
\begin{proof}
To prove this, let us consider the $M$-images of the other standard
basis vectors $\ve_i$, $1\le i\le d-1$.  First $M(\ve_i)$ has to be
orthogonal (in the Euclidean sense) to $\ve_d$, this is so since both
$\ve_d+M(\ve_i)$ and $\ve_d-M(\ve_i)$ has to be lightlike, see
Figure~\ref{fig-eme}.

\begin{figure}
\begin{center}
\begin{tikzpicture}[scale=0.75]
\tikzstyle{vektor} = [->,very thick, >=stealth]
\tikzstyle{cimke}=[fill opacity=0.9, inner sep=0.5, fill=white]
\pgfmathsetmacro{\md}{2}
\pgfmathsetmacro{\mi}{3}  
\pgfmathsetmacro{\size}{2.5}
\pgfmathsetmacro{\d}{1}

\begin{scope}[shift={(-4.5,0)}]
\coordinate (o) -- (0,0);
\draw (-\mi,-\d)-- (-\mi,\md-\d)-- (\mi,\md+\d)--(\mi,\d);
\draw[red] (-\size,-\size)--(\size,\size) (-\size,\size)--(\size,-\size);
\draw[vektor] (o)-- (\mi,\d)  node[below,xshift=5]{$M(\bar e_i)$};
\draw[vektor] (o)-- (\mi,\md+\d) node[above,cimke]{$\bar e_d+M(\bar e_i)$};
\draw[vektor] (o)-- (-\mi,\md-\d) node[above,cimke]{$\bar e_d-M(\bar e_i)$};
\draw[vektor] (o) -- (-\mi,-\d);
\draw[vektor] (o)-- (0,\md) node[above,cimke]{$\bar e_d=M(\bar e_d)$};
\end{scope}
\draw (0,-\size*1.5) -- (0,\size*1.75);
 \begin{scope}[shift={(3.5,0)}]
\coordinate (o) -- (0,0);
\draw[red] (-\size,-\size)--(\size,\size) (-\size,\size)--(\size,-\size);
\draw (o) circle (2);
\draw[vektor] (o)-- (\mi,0) node[below]{$M(\bar e_i)$};
\draw[vektor] (o)-- (\mi,\md) node[above, xshift=10,cimke]{$\bar e_d+M(\bar e_i)$};
\draw[vektor] (o)-- (0,\md) node[above,cimke]{$\bar e_d=M(\bar e_d)$};
\draw (\mi,0)-- (\mi,\md)-- (0,\md) (\md,\md)--(\md,0);
\end{scope}
\end{tikzpicture}
\caption{\label{fig-eme} Illustration for the proof of Lemma~\ref{lem-a}}
\end{center}
\end{figure}
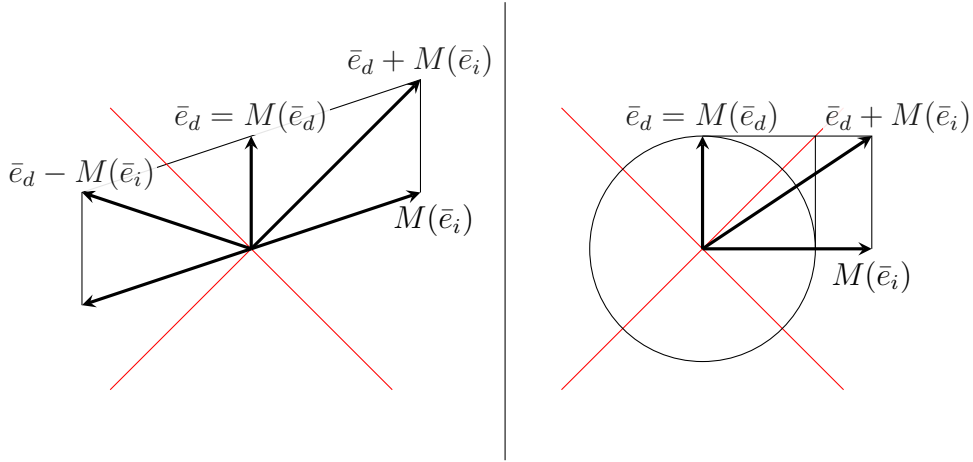

The Euclidean length of $M(\ve_i)$ also has to be $1$ since
$\ve_d+M(\ve_i)$ is lightlike, see
Figure~\ref{fig-eme}.  

Finally, $M(\ve_i)$ is orthogonal to $M(\ve_j)$ if $1\le i<j<d$. If
$d=2$, there is nothing to be proved. If $d\ge3$, the $M$-image of the
lightlike vector $\ve_d+\frac{3}{5}\ve_i+\frac{4}{5}\ve_j$ has to be
lightlike.  By the linearity of $M$, this $M$-image is
$\ve_d+\frac{3}{5}M(\ve_i)+\frac{4}{5}M(\ve_j)$, which is lightlike
iff $M(\ve_i)$ is orthogonal to $M(\ve_i)$. 

These facts imply that $M$ is an isometry on the space part of $\Q^d$
fixing $\ve_d$. Hence $M$ is a space isometry.
\end{proof}

\begin{lemma}\label{lem-llin}
Let $d\ge3$ and assume \ax{AxEField}. Any linear bijection $A$ from $\Q^d$
to $\Q^d$ taking lightlike vectors to lightlike vectors is a Lorentz
transformation composed by a dilation.\footnote{Lemma~\ref{lem-llin}
  can also be proved by using the Alexandrov-Zeeman theorem
  generalized over ordered fields, see \cite{VKK} or \cite{Pambuccian}.}
\end{lemma}

\begin{proof}
Let us first note that $A$ takes timelike vectors to timelike
ones. This is so since timelike vectors can be defined by the
following property: $\vt$ is a timelike vector iff $\vt\neq\vo$ and
for any lightlike vector $\vp$ there is another lightlike vector $\vq$
such that $\vp+\vq=\lambda\cdot\vt$ for some $0\neq
\lambda\in\Q$. This fact can be proved from \ax{AxEField} since $d\ge3$.

Hence $A(\ve_d)$ is timelike. Timelike vector $\ve_d$ can be
transformed to timelike vector $A(\ve_d)$ by a Lorentz boost
(hyperbolic rotation) $B$, space isometry $S$, and dilation $D$. Let
$M$ be  $A\comp (B\comp S\comp
D)^{-1}$.  $M(\ve_d)=\ve_d$ and $M$ takes lightlike vectors to
lightlike ones (by the properties of its decomposition). By Lemma
\ref{lem-a}, we have that $M$ is a space isometry.

Thus $A=M\comp B\comp S\comp D$. This completes the proof since $M\comp B\comp S$ is a Lorentz
transformation and $D$ is a dilation.
\end{proof}

\begin{proof}[\colorbox{proofbgcolor}{\textcolor{proofcolor}{Proof of Theorem~\ref{thm-lordf}}}]
Let $m$ and $k$ be two observers and let $\vx\in
\wl_m(k)\cap\wl_m(m)$.  By Lemma~\ref{lem-ll}, we have that $\w_{mk}$
is differentiable at $\vx$ and $[d_{\vx}\w_{mk}]$ is a linear
bijection taking lightlike vectors to lightlike vectors.  Hence, by
Lemma~\ref{lem-llin}, $[d_{\vx}\w_{mk}]$ has to be a Lorentz
transformation $L$ composed by a dilation $D$, i.e.,
$[d_{\vx}\w_{mk}]=L\comp D$.

By Lemma~\ref{lem-sym}, \ax{AxSymt^-} implies that $[d_{\vx}\w_{mk}]$
has the sym-time property. So dilation $D$ has to be the identity map
because of the followings. 

The sym-time property is true for Lorentz transformation $L$, i.e.,
$L(\ve_d)_t=L^{-1}(\ve_d)_t$. Therefore, if $D$ is a nontrivial
dilation, $L\comp D$ does not have the sym-time property. For example,
if $D$ is an enlargement in the decomposition $[d_{\vx}\w_{km}]=L\comp
D$,
\begin{equation*}
 (L \comp D)(\ve_d)_t>L(\ve_d)_t=L^{-1}(\ve_d)_t>(D^{-1}\comp L^{-1})(\ve_d)_t.
\end{equation*}
An
analogous calculation works in the case when $D$ is a shrinking.
Therefore, $D$ has to be the identity map. So
$[d_{\vx}\w_{mk}]$ is a Lorentz transformation as stated.
\end{proof}

\begin{prop}\label{prop-welldf}
Let $d\ge 3$. \ax{GenRel} implies that $\g_m$ defined by \eqref{eq-g}
is a function for all $m\in\Ob$, i.e., $a$ does not depend on the
choice of observer $k\in\ev_m(\vx)$.
\end{prop}
\begin{proof}
Let $k$ and $h$ be observers such that $k,h\in\ev_m(\vx)$. Then, by
Lemma~\ref{lem-df0}, $\w_{mk}$ and $\w_{mh}$ are functions
differentiable at $\vx$. Let $\vy$ be $\w_{mk}(\vx)$. Since
$k,h\in\ev_m(\vx)$, we have that $k,h\in\ev_k(\vy)$, i.e.,
$\vy\in\wl_k(h)\cap\wl_k(k)$. Therefore, by Theorem~\ref{thm-lordf},
$[d_{\vy}\w_{kh}]$ is a Lorentz transformation, i.e., it preserves the
Minkowski metric. Hence
\begin{equation*}
\mu\big([d_{\vx}\w_{mh}](\vv),[d_{\vx}\w_{mh}](\vw)\big)=\mu\big([d_{\vx}\w_{mk}](\vv),[d_{\vx}\w_{mk}](\vw)\big)
\end{equation*}
for all $\vv,\vw\in\Q^d$ since $[d_{\vx}\w_{mh}]=[d_{\vx}\w_{mk}]\comp [d_{\vy}\w_{kh}]$.
\end{proof}

\begin{proof}[Proof of Theorem~\ref{thm-c0}]
Since the ordered field reduct $\langle \Q,+,\cdot,\le \rangle$ of
$\G$ and $M(\G)$ is the same, \ax{AxEField} is valid in $M(\G)$.

Axiom \ax{AxCDiff} or \ax{AxC^n} for any $n\ge 1$ contains that
$\w_{mk}$ is a function. So axiom \ax{AxFn\psi} is valid in $M(\G)$.
By Proposition~\ref{prop-welldf}, $\g_m$ is a function.  Therefore,
\ax{AxFn\g} is also valid in $M(\G)$.

Let $m$, $k$ and $h$ be observers. We have that $\w_{mm}=\Id_{\dom
  \w_{mm}}$, $\w_{mk}= \w_{km}^{-1}$ and $\w_{mk}\comp\w_{kh}\subseteq
\w_{mh}$ by the definition of worldview transformation and the fact
that they are functions (by axiom \ax{AxCDiff} or \ax{AxC^n}).
Therefore, \ax{AxCom\psi} is valid in $M(\G)$.

\ax{AxCDiff} is $Tr\big(\ax{AxCDiff\psi}\land\ax{AxFn\psi}\big)$;
axiom Axiom \ax{AxC^n} is
$Tr\big(\ax{AxC^n\psi}\land\ax{AxFn\psi}\big)$.  Hence
\ax{AxCDiff\psi} (\ax{AxC^n\psi}) is valid in $M(\G)$ iff \ax{AxCDiff}
(\ax{AxC^n}) is valid in $\G$.

Axiom \ax{AxFull\g} is valid in $M(\G)$ because of the followings: by
axiom \ax{AxThExp^-_{00}} there is an observer $k$ such that
$\W(m,k,\vx)$ for all $m\in\Ob$ and $\vx\in\dom \w_{mm}$; therefore,
$\g_m$ is defined on $\dom \w_{mm}$ for all  observer
$m$.

Axiom \ax{AxC^0\g} follows from \ax{AxC^0\g_m} by
Proposition~\ref{prop-welldf}.  Axiom \ax{AxC^n\g_m} is
$Tr\big(\ax{AxC^n\g}\big)$ for all $n\ge 1$.  Therefore,
\ax{AxC^n\g} is valid in $M(\G)$ iff \ax{AxC^n\g_m} is valid in $\G$.

To prove that axiom \ax{AxLor\g} is valid in $M(\G)$, let $m$ be an
 observer and let $\vx\in \dom\g_m$. By the definition of $\g$,
$\g_m(\vx)(\vv,\vw)=\mu\big([d_{\vx}\w_{mk}](\vv),[d_{\vx}\w_{mk}](\vw)\big)$
for some  observer $k$ for which $\vx\in \wl_m(k)$.  Since
$\vx\in \wl_m(k)$, linear map $[d_{\vx}\w_{mk}]$ exists by
Lemma~\ref{lem-df0}. So we can choose $[d_{\vx}\w_{mk}]$ to be $L$ in
\ax{AxLor\g}. Hence \ax{AxLor\g} is valid in $M(\G)$.

To prove that axiom \ax{AxCom\g} is valid in $M(\G)$, let $m$ and $h$
be  observers and let $\vx\in\dom\g_m\cap\dom\w_{mh}$.  We have
to show that
$\g_m(\vx)(\vv,\vw)=\g_h\big(\w_{mh}(\vx)\big)\big([d_{\vx}\w_{mh}](\vv),[d_{\vx}\w_{mh}](\vw)\big)$
for all $\vv,\vw\in\Q^d$.  Since $\vx\in \dom \g_m$, there is an 
observer $k$ in the event $\ev_m(\vx)$. By the definition of $\g$,
\begin{equation*}
\g_m(\vx)(\vv,\vw)=\mu\big([d_{\vx}\w_{mk}](\vv),[d_{\vx}\w_{mk}](\vw)\big)
\end{equation*}
and
\begin{multline*}
\g_h\big(\w_{mh}(\vx)\big)\big([d_{\vx}\w_{mh}](\vv),[d_{\vx}\w_{mh}](\vw)\big)=\\\mu\big([d_{\w_{mh}(\vx)}\w_{hk}]([d_{\vx}\w_{mh}](\vv)),[d_{\w_{mh}(\vx)}\w_{hk}]([d_{\vx}\w_{mh}](\vw))\big).
\end{multline*}
So it is enough to show that
$[d_{\vx}\w_{mk}]=[d_{\vx}\w_{mh}]\comp [d_{\w_{mh}(\vx)}\w_{hk}]$, which is
true by chain rule, see Theorem~\ref{thm-cr}.

Finally, we show that \ax{Cont_\LM} is valid in $M(\G)$.  By
formula induction, it is easy to prove that $M(\G)\models \varphi$ iff
$\G\models Tr(\varphi)$. For example, $M(\G)\models
\g(i,\vx,\vv,\vw,a)$ holds iff 
\begin{equation*}
\G \models \exists j \enskip \Ob(j) \land \W(i,j,\vx)\land
\mu\big([d_{\vx}\w_{ij}](\vv),[d_{\vx}\w_{ij}](\vw)\big)=a
\end{equation*}
which holds iff $\G\models Tr\big(\g(i,\vx,\vv,\vw,a)\big)$ by the
definitions of $M(\G)$ and $Tr$; and $M(\G)\models \exists i\enskip
\varphi$ iff there is an $a\in\Q\cup\I$  such
that $M(\G)\models \varphi[a]$ iff there is an $a\in\Q\cup\B$ such
that $\G\models Tr(\varphi[a])$ iff $\G\models \exists i \enskip
Tr(\varphi)$.

Let $\varphi(x,\vy)$ be a formula in the language of \ax{LorMan} such
that $x$ is a free variable of $\varphi$ of sort $\Q$ and all the
other free variables of $\varphi$ are amongst $\vy$. Quantity $a$ is
in the set defined by $\varphi$ and parameter $\vp$ iff $M(\G)\models
\varphi[a,\vp]$. By the above, this is equivalent to that $\G\models
Tr(\varphi)[a,\vp]$.  This means that $a$ is in the set defined by
$Tr(\varphi)$ using $\vp$ as parameters.

By the construction of model $M(\G)$ we have that the structure
$\langle \Q,+,\cdot,\le\rangle$ of quantities in $\G$ and $M(\G)$ is
the same.  Consequently, the supremum of the set defined by $\varphi$
by parameters $\vp$ and the supremum of the set defined by
$Tr(\varphi)$ by parameters $\vp$ is the same.  This means that
$\ax{AxSup_{Tr(\varphi)}}\in \ax{CONT_\LG}$ implies
$\ax{AxSup_\varphi}\in\ax{CONT_\LM}$.  Hence, \ax{CONT_\LM} is true in
$\M(\G)$ since \ax{CONT_\LG} is true in $\G$.
\end{proof}

\begin{lemma}\label{lem-deftlc}
Assume \ax{AxEField}. Let $\vx,\vy,\vv,\vw\in\Q^d$ such that
$\vy-\vx$, $\vv$ and $\vw$ are definable timelike vectors for which
$y_t>x_t$, $v_t>0$ and $w_t>0$. Then there is a continuously
differentiable definable timelike curve $\gamma$ such that
$\gamma(0)=\vx$, $\gamma(1)=\vy$, $\gamma'(0)=\alpha\vv$ and
$\gamma'(1)=\beta\vw$ for some positive $\alpha$ and $\beta$.
Moreover, there is a positive $\delta$ such that $|\gamma_s'(t)|\le
(1-\delta)|\gamma'_t(t)|$ for all $t\in [0,1]$.
\end{lemma}
\begin{proof}
We can assume without loosing generality that $\vx= \vo$ and
$\vy=\langle 1,0,\ldots, 0\rangle$ because by a composition of a
definable translation, a definable Lorentz transformation, and a
definable scaling we can map $\vx$ to $\vo$ and $\vy$ to $\langle
1,0,\ldots, 0\rangle$ without changing the required properties of
$\gamma$.

Let 
\begin{equation*}
\gamma(t)\de \left\langle t,
\frac{\vv_s}{v_t}(t^3-2t^2+t)+\frac{\vw_s}{w_t}(t^3-t^2)\right\rangle,
\end{equation*}
for all $t\in\Q$. 
Then 
\begin{equation*}
\gamma'(t)= \left\langle t,
\frac{\vv_s}{v_t}(3t^2-4t+1)+\frac{\vw_s}{w_t}(3t^2-2t)\right\rangle.
\end{equation*}
It is straightforward to verify that $\gamma(0)=\vo$,
$\gamma(1)=\langle 1,0,\ldots, 0\rangle$, $\gamma'_s(0)=\vv_s/v_t$,
$\gamma_s'(1)=\vw_s/w_t$, $\gamma'_t(t)=1$ for all $t\in[0,1]$. Hence
$\gamma'(0)=\alpha \vv$ and $\gamma'(1)=\beta\vw$ for $\alpha=v_t$ and
$\beta=w_t$, which are positive quantities. 

It is also clear that $\gamma$ is continuously differentiable.  Let u
now show that $\gamma'$ is a timelike vector for all $t\in[0,1]$.

\begin{multline*}
|\gamma'_s(t)|=\left|\frac{\vv_s}{v_t}(3t^2-4t+1)+\frac{\vw_s}{w_t}(3t^2-2t)\right|\\
\le \max\left(\left|{\vv_s}/{v_t}\right|,\left|{\vw_s}/{w_t}\right|\right)(|3t^2-4t+1| +|3t^2-2t|)\\\le\max\left(\left|{\vv_s}/{v_t}\right|,\left|{\vw_s}/{w_t}\right|\right)
\end{multline*}
since $|3t^2-4t+1|+|3t^2-4t+1|<1$ if $t\in [0,1]$. Consequently, there
is a $\delta>0$ such that $|\gamma_s(t)|<1-\delta$ because
$|\vv_s|<|v_t|$, $|\vw_s|<|w_t|$.  Therefore,
\begin{equation*}
|\gamma'_s(t)|<(1-\delta)|\gamma'_t(t)| \text{ for all }
t\in[0,1]
\end{equation*}
since 
$|\gamma'_t(t)|=1$ for all $t\in[0,1]$.
\end{proof}

\section{Concluding Remarks}
We have introduced several FOL axiom systems \ax{GenRel^n} for general
relativity and showed that they are complete with respect to
Lorentzian manifolds having the corresponding smoothness properties,
see Theorem~\ref{thm-c0}. From \cite{logroad}, we recalled our FOL
definition of timelike geodesic formulated in the language of
\ax{GenRel}, see \eqref{eq-tlgeod}, and justified this definition by
 showing that our FOL definition coincides with
the usual notion of geodesic over the field $\Reals$ of real
numbers, see Theorem~\ref{thm-geod}. Since all the other key notions of GR, such as curvature or
Riemannian tensor field, are definable from timelike geodesics, we can
also define all these notions in \ax{GenRel}.

A future task is building our axiomatic hierarchy of relativity
theories further, i.e., finding natural axiom systems similar to
\ax{GenRel} which are complete with respect to certain spacetime
classes, such as black holes, cosmological spacetimes, etc. For
example, see \cite{NSz} for an axiom capturing Malament--Hogarth
spacetimes in the language of \ax{GenRel}.

Another task is taking alternative axiom systems for general
relativity (possibly in a completely different language, such as the
language of causality, see e.g., \cite{KP67}) and logically compare
these axiom systems to \ax{GenRel}, e.g., interpreting one in the
another or proving their definitional equivalence using the techniques
of \cite{AN-CompTh} and \cite{Mphd}.  This task is a part of the so
called conceptual analysis of the relativity theory and it helps to
understand the roles and connections of the possible basic concepts of
the theory.

A third task is taking some (preferably surprising) predictions of GR
and finding a minimal set of (natural) axioms implying this
prediction. This task is a kind of answering why-type questions of
relativity theory, see e.g., \cite{wqp}.  For this kind of reverse
analysis in SR, see \cite[\S 3.4]{pezsgo}, \cite{logroad} on
impossibility of faster than light motion, \cite{twp}, \cite{clp},
\cite{Szphd} on the twin paradox.

Doing research in any of the three tasks above will lead us to a
deeper (more structured, axiomatic) understanding of the theory of
GR.


\bibliography{LogRelBib}

\begin{thebibliography}{10}

\bibitem{pezsgo}
H.~Andr{\'e}ka, J.~X. Madar{\'a}sz, {and}~I. N{\'e}meti, with contributions
  from:~A. Andai, G.~S{\'a}gi, I.~Sain, and Cs. T{\H o}ke.
\newblock {\it On the logical structure of relativity theories}.
\newblock Research report, Alfr{\'e}d R{\'e}nyi Institute of Mathematics,
  Hungar. Acad. Sci., Budapest, 2002.
\newblock http://www.math-inst.hu/pub/algebraic-logic/Contents.html.

\bibitem{logst}
H.~Andr{\'e}ka, J.~X. Madar{\'a}sz, and I.~N{\'e}meti.
\newblock Logic of space-time and relativity theory.
\newblock In M.~Aiello, I.~Pratt-Hartmann, and J.~van Benthem, editors, {\em
  Handbook of spatial logics}, pages 607--711. Springer-Verlag, Dordrecht,
  2007.

\bibitem{logroad}
H.~Andr{\'e}ka, J.~X. Madar{\'a}sz, I.~N{\'e}meti, and G.~Sz{\'e}kely.
\newblock A logic road from special relativity to general relativity.
\newblock {\em Synthese}, 186(3):633--649, 2012.

\bibitem{AN-CompTh}
H.~{Andr{\'e}ka} and I.~{N{\'e}meti}.
\newblock {Comparing theories: the dynamics of changing vocabulary. A
  case-study in relativity theory}.
\newblock {\em arXiv:1307.1885}, 2013.

\bibitem{ax}
J.~Ax.
\newblock The elementary foundations of spacetime.
\newblock {\em Found. Phys.}, 8(7-8):507--546, 1978.

\bibitem{Basri}
S.~A. Basri.
\newblock {\em A deductive theory of space and time}.
\newblock Studies in logic and the foundations of mathematics. North-Holland
  Pub. Co., 1966.

\bibitem{GlobLorGeom}
J.~K. Beem, P.~E. Ehrlich, and K.~L. Easley.
\newblock {\em Global Lorentzian Geometry}.
\newblock Chapman and Hall/CRC Pure and Applied Mathematics Series. Marcel
  Dekker Incorporated, 1996.

\bibitem{benda}
T.~Benda.
\newblock A formal construction of the spacetime manifold.
\newblock {\em J. Phil. Logic}, 37(5):441--478, 2008.

\bibitem{CK}
C.~C. Chang and H.~J. Keisler.
\newblock {\em Model theory}.
\newblock North-Holland Publishing Co., Amsterdam, 1990.

\bibitem{dinverno}
R.~d'Inverno.
\newblock {\em Introducing {E}instein's relativity}.
\newblock Oxford University Press, New York, 1992.

\bibitem{FriFOM1}
H.~Friedman.
\newblock On foundational thinking 1.
\newblock Posting in FOM (Foundations of Mathematics) Archives, www.cs.nyu.edu,
  January 20, 2004.

\bibitem{FriFOM2}
H.~Friedman.
\newblock On foundations of special relativistic kinematics 1.
\newblock Posting No 206 in FOM (Foundations of Mathematics) Archives,
  www.cs.nyu.edu, January 21, 2004.

\bibitem{goldblatt}
R.~Goldblatt.
\newblock {\em Orthogonality and spacetime geometry}.
\newblock Springer-Verlag, New York, 1987.

\bibitem{Hawking-Ellis}
S.~W. Hawking and G.~F.~R. Ellis.
\newblock {\em The large scale structure of space-time}.
\newblock Cambridge University Press, London, 1973.
\newblock Cambridge Monographs on Mathematical Physics, No. 1.

\bibitem{Hodges}
W.~Hodges.
\newblock {\em Model theory}.
\newblock Cambridge University Press, Cambridge, 1993.

\bibitem{KP67}
E.~H. Kronheimer and R.~Penrose.
\newblock On the structure of causal spaces.
\newblock {\em Proc. Cambridge Philos. Soc.}, 63:481--501, 1967.

\bibitem{Latzer}
R.~W. Latzer.
\newblock Nondirected light signals and the structure of time.
\newblock {\em Synthese}, 24(1-2):236--280, 1972.

\bibitem{Mphd}
J.~X. Madar{\'a}sz.
\newblock {\em Logic and Relativity (in the light of definability theory)}.
\newblock PhD thesis, E{\"o}tv{\"o}s Lor{\'a}nd Univ., {B}udapest, 2002.
\newblock http://www.math-inst.hu/pub/algebraic-logic/Contents.html.

\bibitem{twp}
J.~X. Madar{\'a}sz, I.~N{\'e}meti, and G.~Sz{\'e}kely.
\newblock Twin paradox and the logical foundation of relativity theory.
\newblock {\em Found. Phys.}, 36(5):681--714, 2006.

\bibitem{mundy-oaomstg}
B.~Mundy.
\newblock Optical axiomatization of {M}inkowski space-time geometry.
\newblock {\em Philos. Sci.}, 53(1):1--30, 1986.

\bibitem{mundy-tpcomg}
B.~Mundy.
\newblock The physical content of {M}inkowski geometry.
\newblock {\em The British Journal for the Philosophy of Science},
  37(1):25--54, 1986.

\bibitem{NSz}
P.~N\'emeti and G.~Sz\'ekely.
\newblock Existence of faster than light signals implies hypercomputation
  already in special relativity.
\newblock In S.~B. Cooper, A.~Dawar, and B.~L\"owe, editors, {\em How the World
  Computes}, volume 7318 of {\em Lecture Notes in Computer Science}, pages
  528--538. Springer Berlin Heidelberg, 2012.

\bibitem{ONeill}
B.~O'Neill.
\newblock {\em Semi-Riemannian Geometry With Applications to Relativity}.
\newblock Pure and Applied Mathematics. Elsevier Science, 1983.

\bibitem{Pambuccian}
V.~Pambuccian.
\newblock Alexandrov-{Z}eeman type theorems expressed in terms of definability.
\newblock {\em Aequationes Math.}, 74(3):249--261, 2007.

\bibitem{LRR11}
E.~Poisson, A.~Pound, and I.~Vega.
\newblock The motion of point particles in curved spacetime.
\newblock {\em Living Reviews in Relativity}, 14(7), 2011.

\bibitem{Robb}
A.~A. Robb.
\newblock {\em A Theory of Time and Space}.
\newblock Cambridge University Press, Cambridge, 1914.

\bibitem{Robb2}
A.~A. Robb.
\newblock {\em Geometry of Time and Space}.
\newblock Cambridge University Press, Cambridge, 1936.

\bibitem{schutz}
J.~W. Schutz.
\newblock {\em Foundations of special relativity: kinematic axioms for
  {M}inkowski space-time}.
\newblock Springer-Verlag, Berlin, 1973.

\bibitem{schutz-aasfmst}
J.~W. Schutz.
\newblock An axiomatic system for {M}inkowski space-time.
\newblock {\em J. Math. Phys.}, 22(2):293--302, 1981.

\bibitem{Schu}
J.~W. Schutz.
\newblock {\em Independent axioms for {M}inkowski space-time}.
\newblock Longoman, London, 1997.

\bibitem{suppes-sopitposat}
P.~Suppes.
\newblock Some open problems in the philosophy of space and time.
\newblock {\em Synthese}, 24:298--316, 1972.

\bibitem{Szabo}
L.~E. Szab{\'o}.
\newblock Empirical foundation of space and time.
\newblock In M.~Su{\'a}rez, M.~Dorato, and M.~R{\'e}dei, editors, {\em EPSA07:
  Launch of the European Philosophy of Science Association}, pages 251--266.
  Springer, 2010.

\bibitem{mythes}
G.~Sz{\'e}kely.
\newblock A first order logic investigation of the twin paradox and related
  subjects.
\newblock Master's thesis, E{\"o}tv{\"o}s Lor{\'a}nd Univ., {B}udapest, 2004.

\bibitem{Szphd}
G.~Sz{\'e}kely.
\newblock {\em First-Order Logic Investigation of Relativity Theory with an
  Emphasis on Accelerated Observers}.
\newblock PhD thesis, E{\"o}tv{\"o}s Lor{\'a}nd Univ., {B}udapest, 2009.
\newblock http://www.renyi.hu/~turms/phd.pdf.

\bibitem{clp}
G.~Sz{\'e}kely.
\newblock A geometrical characterization of the twin paradox and its variants.
\newblock {\em Studia Logica}, 95:161--182, 2010.

\bibitem{wqp}
G.~Sz{\'e}kely.
\newblock On why-questions in physics.
\newblock In A.~M{\'a}t{\'e}, M.~R{\'e}dei, and F.~Stadler, editors, {\em The
  Vienna Circle in Hungary, Wiener Kreis und Ungarn}, pages 181--189. Springer,
  Wien, 2011.

\bibitem{VKK}
P.~G. Vroegindewey, V.~Kreinovic, and O.~M. Kosheleva.
\newblock An extension of a theorem of {A}. {D}. {A}leksandrov to a class of
  partially ordered fields.
\newblock {\em Indag. Math.}, 41(3):363--376, 1979.

\bibitem{wald}
R.~M. Wald.
\newblock {\em General relativity}.
\newblock University of Chicago Press, Chicago, 1984.

\end{thebibliography}
\bibliographystyle{plain}

\end{document}